\documentclass[a4paper,11pt]{quantumarticle}
\pdfoutput=1
\usepackage[utf8]{inputenc}
\usepackage[T1]{fontenc}
\usepackage{amssymb,amsmath,amsthm,amsfonts}
\usepackage{mathtools}
\usepackage{stmaryrd}
\usepackage{physics}
\usepackage{hyperref}
\usepackage{bm}
\usepackage[numbers]{natbib}
\usepackage{tikz}
\usepackage{booktabs}
\usepackage{geometry}
\usepackage{float}
\usepackage{graphicx}
\usepackage{quantikz}
\usetikzlibrary{arrows.meta, positioning,calc, shapes.geometric, fit, backgrounds}
\usepackage{circuitikz}
\usepackage{fontawesome5}
\usepackage{standalone}
\def\id{{\rm 1\kern-.22em l}}
\newcommand{\bigo}[1]{O\!\left(#1\right)}      
\newcommand{\bigt}[1]{\Theta\!\left(#1\right)}   
\newcommand{\C}{\mathbb{C}}

\newcommand{\projj}[1]{\ketbra{#1}{#1}}
\newcommand{\BB}{\mathcal{B}}
\newcommand{\PP}{\mathcal{P}}
\newcommand{\LL}{\mathcal{L}}

\newcommand{\HH}{\mathcal{H}}
\newcommand{\UU}{\mathcal{U}}

\newtheorem{thm}{Theorem}
\newtheorem{propo}[thm]{Proposition}
\newtheorem{lem}[thm]{Lemma}
\newtheorem{deff}[thm]{Definition}

\newtheorem{remark}[thm]{Remark}
\newtheorem{cor}[thm]{Corollary}

\begin{document}

\title{Quantum and Classical Memory for Storage and Retrieval of Quantum Measurements}

\author{Filip Turoboś}
\email{filip.turobos@p.lodz.pl}
\affiliation{Institute of Theoretical and Applied Informatics, Polish Academy of Sciences, ul. Ba{\l}tycka 5, 44-100 Gliwice, Poland}
\orcid{0000-0002-5782-6159}

\author{Ludmiła Marcinkowska}
\email{lmarcinkowska@iitis.pl}
\affiliation{Institute of Theoretical and Applied Informatics, Polish Academy of Sciences, ul. Ba{\l}tycka 5, 44-100 Gliwice, Poland}
\orcid{0009-0002-0935-3921}

\author{Zbigniew Puchała}
\email{zpuchala@iitis.pl}
\affiliation{Institute of Theoretical and Applied Informatics, Polish Academy of Sciences, ul. Ba{\l}tycka 5, 44-100 Gliwice, Poland}
\orcid{0000-0002-4739-0400}

\maketitle
\onecolumn
\raggedbottom

\begin{abstract}
We study the persistent quantum and classical memory required for exact heralded storage and retrieval of quantum measurements. Given only black-box access to the measurement channel, we construct an exact compression of the port-based program into an outcome-type record and a quantum register, with total dimension polynomial in the query count. For every universal retrieval processor we prove the bounds $Q+C\ge\log_2 n$ and $2Q+C\ge\log_2\bigl((n-1)d^2+1\bigr)$ on the numbers $Q$ of qubits and $C$ of bits of the program, for $n$-outcome POVMs on $\C^d$ with $d\ge2$. For $n\ge2$, no finite classical memory removes the need for quantum memory: any such processor requires a quantum program of dimension at least $d$, at that dimension its classical alphabet must have at least $n$ symbols, and both minima are attained by a single-query protocol with success $1/d^2$. Combining direct-program compression with a converse from approximate measurement programming gives logarithmic total-memory scaling in $1/(1-p)$ as the uniform success probability $p$ approaches one, at fixed $d,n\ge2$. At fixed input dimension and target success below one, the optimal total-memory scaling is logarithmic in $n$. We also construct self-contained programs under branch-decomposition access and compare the quantum and classical costs of all constructions, including where the outcome-type record is kept: for POVMs with $n$ nonzero rank-one effects the direct program admits a purely quantum allocation whose register is never larger than that of the minimal-Naimark unitary program and which succeeds more often whenever $n>d$.
\end{abstract}

\section{Introduction}

Extracting information about an unknown quantum object from a finite number of uses is a basic task of quantum information: a state cannot be reconstructed from one copy~\cite{Aaronson2017Shadow}, state discrimination has sharp limits~\cite{chefles2000quantum}, and an unknown transformation can be queried a finite number of times, stored in a quantum memory and later reproduced on new inputs~\cite{bisio2010optimal}. For measurements, Bisio et al.~\cite{bisio2011quantum} showed that optimal learning of von Neumann measurements may require coherent memory and need not be parallelizable, and port-based teleportation~\cite{ishizaka2009quantum,Pirandola2019,christandl2021asymptotic,yoshida2026one,wang2016higher} led to storage-and-retrieval schemes for projective measurements with average fidelity $1-\bigt{N^{-2}}$~\cite{lewandowska2022storage}. For unitary channels, probabilistic storage and retrieval (PSAR) based on probabilistic port-based teleportation (pPBT)~\cite{ishizaka2008asymptotic,studzinski2017port} achieves exact heralded retrieval with a program polynomial in the number $N$ of uses~\cite{sedlak2019optimal,Bisio2011}.

A POVM channel is not unitary, and the same POVM can be presented to a protocol in operationally inequivalent ways: as a measurement channel~\cite{watrous2018theory}, so that only the statistics $\Tr(M_i\rho)$ are accessible; as a coherent Naimark dilation~\cite{watrous2018theory,Chen2007}; or already decomposed into a mixture of simpler measurements~\cite{oszmaniec2018simulating}, here in the Kotowski-Oszmaniec form~\cite{kotowski2026pretty}, which represents any POVM as a probabilistic ensemble of measurement branches generated by unitaries on a single-qubit extension. Each presentation is an oracle for the same measurement with a different storage cost. We take these oracle models as the setting rather than the result, and ask: how much program memory does exact heralded retrieval of an unknown measurement require, and how does it divide between a coherent register and a classical record? We do not address the construction of dilations or branch decompositions, and the oracle models are not assumed to be interconvertible (Section~\ref{sec:discussion}).

\paragraph{Relation to prior work.}
The utilized tools are well-established in the literature: port-based teleportation~\cite{ishizaka2009quantum,studzinski2017port}, exact compression of permutation-invariant ensembles by discarding multiplicity spaces~\cite{yang2016efficient}, optimal-success PSAR of unitaries~\cite{sedlak2019optimal}, the Kotowski--Oszmaniec simulation theorem~\cite{kotowski2026pretty}, and the approximate-programming lower bound of~\cite{perezgarcia2006optimality}. Exact probabilistic programming of measurements with an inconclusive outcome, for programs supplied as states, was studied for qubit projective measurements by Fiur\'a\v{s}ek and Du\v{s}ek~\cite{fiurasek2004probabilistic}; compression preserving only the statistics of a prescribed set of measurements was studied in~\cite{bluhm2018compression}; the conversion of heralded exact programming into deterministic approximate programming is discussed for unitaries in~\cite{kubicki2019resource}. What is new here is measurement-specific: the compression of the program generated by black-box queries to an arbitrary POVM, the lower bounds on its quantum and classical parts, and the measured encoders for structured access, as listed below. We make no optimality claim for the compressed registers.

\subsection{Contributions}
\begin{enumerate}
\item \textbf{Exact compression of the black-box program.} Under direct access, the program of the universal pPBT protocol is permutation-invariant with classical outcomes and compresses exactly into the type of the outcome string and a quantum register, of total dimension polynomial in $N$ (Proposition~\ref{prop:direct-cq}). The compression mechanism stems from~\cite{yang2016efficient}, and what is new is its application to the pPBT measurement program and the resulting exponent. Discarding the port alignment can destroy exactness: keeping only the type, without re-aligning the register, does not permit exact retrieval by classical post-processing that returns an outcome on every pPBT success (Lemma~\ref{lem:ppbt-rigidity}, Proposition~\ref{prop:type-only}).
\item \textbf{A branch program with finite classical readout.} Under branch access, an exact heralded protocol has a quantum register of dimension $\binom{N+4d^2-1}{N}$, independent of $n$, and a classical record of $2d$ outcome labels, independent of $N$ and of the number of branches; the receiver needs no measurement-dependent relabelling library (Proposition~\ref{prop:finite-readout}); measuring the branch during storage and applying the rank-one encoder gives a program with a smaller register and higher success $q\lambda(N,d)$; with the raw-outcome type stored in the quantum register its classical record is the same $N$-independent table, so the full-unitary program is dominated as an implementation of the measurement and is retained as the reference construction imported from unitary PSAR (Proposition~\ref{prop:measured-branch}, Corollary~\ref{cor:measured-branch-quantum}). Its success probability stays below the non-failure probability $q$ of the decomposition, which can be taken to be $\tfrac18$ for every POVM~\cite{kotowski2026pretty}, or $\tfrac1d$ with the projective simulation of~\cite{oszmaniec2018simulating} (Remark~\ref{rem:spectral-branch}).
\item \textbf{Lower bounds for every access model.} For a processor that retrieves every $n$-outcome POVM on $\C^d$ exactly with positive success, the program dimension is at least $n$ (Proposition~\ref{prop:label-memory-lower}). The numbers of qubits $Q$ and bits $C$ satisfy $2Q+C\ge\log_2\bigl((n-1)d^2+1\bigr)$ (Proposition~\ref{prop:param-lower}), for $n\ge2$ the quantum register alone has dimension at least $d$ for every finite classical budget (Proposition~\ref{prop:qdim-lower}), and the memory grows at least logarithmically in $1/(1-p)$ as the success probability $p\to1$ (Proposition~\ref{prop:success-lower}, Corollary~\ref{cor:success-all}). Together with the direct construction, the total memory is $\Theta(\log n)$ at fixed $d$ and target success $p_0<1$ (Corollary~\ref{cor:direct-outcome-scaling}) and $\Theta\bigl(\log\frac1{1-p}\bigr)$ at fixed $d,n\ge2$ (Remark~\ref{rem:success-lower}).
\item \textbf{Quantum/classical comparison.} We compare direct, Naimark and branch storage with quantum and classical memory reported separately and with the fixed resource stated in each comparison (Section~\ref{comparison}). Among the constructions compared here, direct pPBT attains the largest success probability at fixed query count: strictly higher than the full-unitary constructions with dilation dimension greater than $d$, and equal when the dilation dimension is $d$. Its type record can be kept classical or absorbed into the quantum register, and for POVMs with $n$ nonzero rank-one effects the purely quantum direct program is never larger than the minimal-Naimark register and succeeds more often whenever $n>d$ (Proposition~\ref{prop:direct-rank1}, Table~\ref{tab:benchmark}); measuring the dilation readout during storage extends this dominance over the full-unitary Naimark program to targets of any rank (Proposition~\ref{prop:measured-naimark}). What the structured routes change is the degree of the polynomial memory and its quantum/classical division. The unitary-PSAR success and register formulas are inherited from~\cite{sedlak2019optimal}; we do not claim their optimality for POVM retrieval.
\end{enumerate}

\section{Preliminaries}

Throughout the paper, $d$ denotes the dimension of the input Hilbert space $\mathbb{C}^d$. We use the standard notation $\BB(\C^{d})$ for the space of linear operators and $\Omega(\mathbb{C}^d)$ for the set of density operators on $\mathbb{C}^d$. The symbol $\UU(D)$ denotes the unitary group\footnote{We will utilize symbol $\UU(\C^d\otimes\C^e)$ when the tensor factorization matters.} on $\C^D$. We denote by $\mathcal{I}_d$ the identity channel on $\C^d$. Matrix trace is denoted by $\mathrm{Tr}$ and for any $d$ the $d$-dimensional identity matrix is denoted by $\id_d$. 

\begin{deff}[Positive operator-valued measure]\label{def:povm}
An $n$-outcome positive operator-valued measure (POVM) on $\mathbb C^d$ is a tuple
\begin{equation}
    \mathbf M=(M_1,\ldots,M_n)
\end{equation}
of positive semidefinite operators $M_{i}\in\BB(\C^{d})$ satisfying
\begin{equation}
    M_i\ge 0,\qquad\sum_{i=1}^{n}M_i=\id_d.
\end{equation}
\end{deff}

Quantum channels are completely positive trace-preserving (CPTP) maps, and a quantum instrument with outcome set $\mathcal X$ is a collection of completely positive maps $\{\mathcal T_x\}_{x\in\mathcal X}$ whose sum is trace preserving. Throughout this paper, measurements will be represented by their associated quantum-to-classical channels.

\begin{deff}[POVM channel]
    The quantum-to-classical channel associated with the $n$-outcome POVM (the \emph{POVM channel}) is $\PP_{\mathbf{M}}\!:\,\BB(\C^{d})\to\BB(\C^{n})$, such that
    \begin{equation}
        \PP_{\mathbf{M}}(\rho)\;=\;\sum_{i=1}^{n}\mathrm{Tr}(M_{i}\rho)\projj{i}.
    \end{equation}
\end{deff}

\begin{deff}[Dynamical matrix \cite{watrous2018theory,bengtsson2017, jamiolkowski1972linear,choi1975completely}]
Let $\Phi:\BB(\C^{d})\to\BB(\C^{d'})$ be a linear map. Its dynamical matrix is
\begin{equation}
    J_{\Phi}=\sum_{i,j=1}^{d}\Phi(|i\rangle\langle j|)\otimes|i\rangle\langle j|\in\BB(\C^{d'}\otimes\C^{d}).
\end{equation}
The map $\Phi$ is completely positive if and only if $J_{\Phi} \geq 0$.
\end{deff}

For the POVM channel $\PP_{\mathbf M}$, the dynamical matrix is
\begin{equation}
    J_{\PP_{\mathbf M}} =\sum_{i=1}^{n}  \projj{i} \otimes M_i^{T}\in\BB(\C^{n}\otimes\C^{d}),
\end{equation}
 with total dimension $dn$, where $\{M_i\}_{i=1}^n$ are the POVM effects acting on $\mathbb{C}^d$ and $\{\ket{i}\}_{i=1}^n$ forms an orthonormal basis representing the classical measurement outcomes in $\mathbb{C}^n$. The normalized Choi state is $\rho_{\mathrm{Choi}} =\frac{1}{d}J_{\PP_{\mathbf M}}.$

A different useful representation is obtained via Naimark dilation~\cite{naimarkblessedpaper}; the version below is Theorem~1 of Chen et al.~\cite{Chen2007} (see also~\cite{watrous2018theory}).

\begin{thm}[Naimark dilation]\label{thm:naimark}
    For every $n$-outcome POVM $\mathbf{M}$ on $\C^{d}$  there exists a Hilbert space $\C^{D_{\mathrm{Nai}}}$ with $D_{\mathrm{Nai}} \ge d$, an isometry $V:\C^{d}\to\C^{D_{\mathrm{Nai}}}$, and a projective partition $\{\Pi_{i}^{(0)}\}_{i=1}^{n}$ of $\id_{D_{\mathrm{Nai}}}$ such that
    \begin{equation}
        M_{i}=V^{\dagger}\Pi_{i}^{(0)}V ,\qquad i=1,\dots,n.
    \end{equation}
    Equivalently, one can write $V=U_{\mathbf{M}}\,\eta_{\mathrm{Nai}}$ for some unitary $U_{\mathbf{M}}\in\UU(D_{\mathrm{Nai}})$ and a fixed embedding isometry $\eta_{\mathrm{Nai}}:\C^{d}\hookrightarrow\C^{D_{\mathrm{Nai}}}$.
    
    The dilated space dimension $D_{\mathrm{Nai}}$ is bounded from below by $D_{\min}$ defined as
    \begin{equation}\label{eq:Dmin}
    D_{\min}:=\sum_{i=1}^{n}\operatorname{rank}(M_{i})\;\le\;dn,
\end{equation}
    while the canonical construction achieves dimension $dn$.
\end{thm}

The dilation dimension is realisation-dependent: the canonical tensor-product realisation uses $\C^{D_{\mathrm{Nai}}}=\C^{d}\otimes\C^{n}$, resulting in $D_{\mathrm{Nai}}=dn$ with $\Pi_{i}^{(0)}=\id_{d}\otimes\projj{i}$. The minimal realisation attains bound \eqref{eq:Dmin} by operating on a direct sum of system space and ancilla (see \cite{Chen2007} for technical details), which satisfies $d\le D_{\min}\le dn$ and equals $n$ for POVMs whose $n$ effects are all nonzero and rank one; if only $k$ effects are nonzero and rank one, $D_{\min}=k$. For the minimal-dilation rank-one benchmarks below all $n$ effects are assumed nonzero, so $D_{\min}=n$. Since $\operatorname{rank}M_i\le\operatorname{rank}\Pi^{(0)}_i$ and the readout ranks sum to $D_{\mathrm{Nai}}$, a minimal realisation forces $\operatorname{rank}\Pi^{(0)}_i=\operatorname{rank}M_i$ for every $i$: the readout of a minimal dilation encodes the rank profile of $\mathbf M$, so minimal-dilation constructions presuppose that this profile is a public promise shared by the whole target family (Definition~\ref{def:coherent-oracle}).

In the storage-and-retrieval (SAR) problem the experimenter has $N$ black-box uses of POVM channel $\PP_\mathbf{M}$, encodes the acquired information into a quantum program register $\HH_{\mathrm{prog}}$, and must later reproduce $\PP_{\mathbf{M}}$'s outcome statistics on an arbitrary input $\rho\in\Omega(\C^{d})$ presented at retrieval~\cite{lewandowska2022storage,sedlak2019optimal}. To quantify retrieval accuracy we use the worst-case channel fidelity.

\begin{deff}[Worst-case channel fidelity]
 For two CPTP maps $\Phi_{1},\Phi_{2}:\BB(\C^{d})\to\BB(\C^{d'})$ we define the worst-case (input) channel fidelity
\begin{equation}
F_{\mathrm{ch}}(\Phi_{1},\Phi_{2})\;:=\;
\min_{\rho\in\Omega(\C^{d})}\,F\!\bigl(\Phi_{1}(\rho),\Phi_{2}(\rho)\bigr),
\end{equation}
where $F(\sigma_{1},\sigma_{2})=\|\sqrt{\sigma_{1}}\sqrt{\sigma_{2}}\|_{1}$ is the Uhlmann state fidelity \cite{uhlmann1976transition,jozsa1994fidelity}. 
\end{deff}

\begin{remark}\label{rem:persistent}
    In this work, we count program memory as the dimension of the Hilbert space in which the stored program state resides. This is \emph{persistent} memory, held between storage and retrieval; the decoder of a compressed program may temporarily recreate the full carrier at retrieval, so the bounds below concern storage, not peak workspace or the efficiency of the decoding circuit. For a channel $\Phi: \mathcal{B}(\mathbb{C}^d) \to \mathcal{B}(\mathbb{C}^{d'})$, the universal pPBT-based construction stores its program on a carrier of dimension $(dd')^N$. Because this program is invariant under relabelling the ports (the argument of Lemma~\ref{lem:direct-symmetry} applies to any $\Phi$), it can be compressed to a dimension polynomial in $N$ (Lemma~\ref{lem:inv-compression}), and for POVM channels further still (Section~\ref{sec:direct-compression}).
\end{remark}

\begin{deff}[Exact heralded retrieval]\label{def:exact-retrieval}
A storage-and-retrieval scheme for a family of POVMs consists of a storage procedure that, from $N$ oracle queries to the unknown $\mathbf M$, prepares a program state $\sigma_{\mathbf M}$ on a fixed register, and of a fixed retrieval instrument $\{\mathcal R_{\mathsf s},\mathcal R_{\mathsf f}\}$ acting on the program and the input. When storage is randomized, $\sigma_{\mathbf M}$ denotes the average of the prepared program states over the randomness of the preparation, so that $\mathcal R_{\mathsf s}(\sigma_{\mathbf M}\otimes\rho)$ is the operational successful output. The scheme retrieves $\mathbf M$ \emph{exactly with heralded success} if
\begin{equation}\label{eq:exact-retrieval}
\mathcal R_{\mathsf s}(\sigma_{\mathbf M}\otimes\rho)=s_{\mathbf M}(\rho)\,\PP_{\mathbf M}(\rho)\qquad\text{for every state }\rho,
\end{equation}
with $s_{\mathbf M}(\rho)>0$ the success probability. The program memory is the pair $(Q,C)$ of qubits and bits in which $\sigma_{\mathbf M}$ is stored, counted as in Remark~\ref{rem:persistent}.
\end{deff}

Condition~\eqref{eq:exact-retrieval} is $F_{\mathrm{ch}}(\hat\PP^{\mathrm{cond}}_{\mathbf M},\PP_{\mathbf M})=1$ for the conditional channel $\hat\PP^{\mathrm{cond}}_{\mathbf M}(\rho)=\mathcal R_{\mathsf s}(\sigma_{\mathbf M}\otimes\rho)/s_{\mathbf M}(\rho)$. The constructions of Sections~\ref{direct}, \ref{naimark} and~\ref{kot-osz} fix $N$ and achieve an input-independent success $s_{\mathbf M}(\rho)=p$ given in closed form; the lower bounds of Section~\ref{sec:lower} assume only~\eqref{eq:exact-retrieval} with positive success for every input and every target, and Lemma~\ref{lem:const-success} shows that for a target with a non-scalar effect the success is then automatically input independent. Two regimes are compared below: at fixed $N$ (Section~\ref{comparison}) the quantity of interest is the allocation $(Q,C)$ at the success $p(N)$; in the scaling statements (Corollaries~\ref{cor:direct-outcome-scaling} and~\ref{cor:outcome-scaling}, Remark~\ref{rem:success-lower}) a target success $p_0$ is prescribed and $N$ is the smallest query count attaining it.

\subsection{Port-based teleportation and PSAR.}
In this subsection, $D$ denotes the dimension of the Hilbert space on which the stored unitary acts.

We use probabilistic port-based teleportation, in which retrieval is exact conditioned on a heralded success event.

\begin{thm}[Probabilistic port-based teleportation {\cite{ishizaka2009quantum}}]
For every $D$ and $N$ there exists a two-outcome quantum instrument $\{\mathcal T_D^{(\textsf{s})},\mathcal T_D^{(\textsf{f})}\}$ whose outcomes are labelled by $\textsf{s}$ (success) and $\textsf{f}$ (failure), satisfying
\begin{equation} \label{eq:ppbt-instrument}
    \mathcal T_D^{(\textsf{s})}=\lambda(N,D)\,\mathcal I_D,
\end{equation}
with the failure branch given by the complementary CP map, so that the two branches form a trace-preserving quantum instrument. The input state is transmitted exactly whenever the success outcome $\textsf{s}$ occurs. The success probability, optimal among probabilistic port-based teleportation protocols with $N$ ports, over resource states and measurements~\cite{ishizaka2009quantum,studzinski2017port}, equals
\begin{equation}
    \lambda(N,D)=\frac{N}{N-1+D^2}=\;1-\frac{D^{2}-1}{N-1+D^{2}}\;=\;1-\bigt{\frac{D^{2}}{N}}, \quad N \gg D^2.
\end{equation}
\end{thm}

For a given implementation of pPBT, we also resolve the success outcome by the port $a\in\{1,\dots,N\}$ at which the input appears, and write $\mathcal T_D^{(\textsf{s},a)}:\BB(\C^D)\to\BB\bigl((\C^D)^{\otimes N}\bigr)$ for the resulting sub-normalized map to all $N$ output ports. Reading out port $a$, i.e.\ tracing out the other ports, and summing over $a$ recovers $\mathcal T_D^{(\textsf{s})}$.

The coherent version of pPBT leads to probabilistic storage-and-retrieval (PSAR) of unknown unitary channels.

\begin{thm}[PSAR for unitary channels {\cite{sedlak2019optimal}}] \label{PSAR}

Let $U\in \UU(D)$ be an unknown unitary available through $N$ uses during the storage phase. There exists a probabilistic storage-and-retrieval protocol such that
\begin{enumerate}
\item retrieval is exact conditioned on success,
\item the success probability equals
\begin{equation}
\lambda(N,D)=\frac{N}{N-1+D^2},
\end{equation}
\item the program register of the optimal-success protocol has dimension
\begin{equation} \label{eq:psar-memory}
\dim\HH_{\mathrm{prog}}=\binom{N+D^2-1}{N}=\bigt{N^{D^{2}-1}}.
\end{equation}
\end{enumerate}
\end{thm}

\begin{remark}[Origin of the polynomial memory]\label{proofPoly}
The register of Theorem~\ref{PSAR} rests on Schur-Weyl duality and on the extra structure of unitary families~\cite{sedlak2019optimal,markiewicz2023duality}. Polynomial size needs only that the program is invariant under relabelling the $N$ uses: such a program can be stored without its permutation degrees of freedom (Lemma~\ref{lem:inv-compression}; see~\cite{yang2016efficient} for exact compression of this kind, including a hybrid quantum--classical version). The pPBT program of every channel has this symmetry (Remark~\ref{rem:persistent}), which gives the direct route its polynomial memory (Section~\ref{sec:direct-compression}).
\end{remark}

\begin{lem}[Compression of permutation-invariant programs]\label{lem:inv-compression}
Let a finite group $G$ act unitarily on $\HH$ by $g\mapsto V_g$, and write its isotypic decomposition as $\HH\cong\bigoplus_\lambda\mathcal K_\lambda\otimes\mathcal R_\lambda$ with $V_g\cong\bigoplus_\lambda\id_{\mathcal K_\lambda}\otimes R_\lambda(g)$, where the $R_\lambda$ are pairwise inequivalent irreducible representations on $\mathcal R_\lambda$. Then
\begin{enumerate}
\item every state with $V_g\sigma V_g^\dagger=\sigma$ for all $g\in G$ has the form $\sigma=\bigoplus_\lambda p_\lambda\,\sigma_\lambda\otimes\id_{\mathcal R_\lambda}/\dim\mathcal R_\lambda$;
\item the maps $\mathcal E(\sigma)=\bigoplus_\lambda\Tr_{\mathcal R_\lambda}\bigl(\Pi_\lambda\sigma\Pi_\lambda\bigr)$ and $\mathcal D\bigl(\bigoplus_\lambda\tau_\lambda\bigr)=\bigoplus_\lambda\tau_\lambda\otimes\id_{\mathcal R_\lambda}/\dim\mathcal R_\lambda$, with $\Pi_\lambda$ the isotypic projectors (on an arbitrary operator, $\mathcal D$ first pinches onto the $\lambda$-blocks), are CPTP and independent of $\sigma$, and $\mathcal D\circ\mathcal E(\sigma)=\frac1{|G|}\sum_{g}V_g\sigma V_g^\dagger$. In particular $\mathcal D\circ\mathcal E(\sigma)=\sigma$ for every invariant $\sigma$, and $\mathcal E(\sigma)$ occupies a register of dimension $\sum_\lambda\dim\mathcal K_\lambda$.
\end{enumerate}
For the symmetric group $S_m$ permuting the factors of $(\C^D)^{\otimes m}$, Schur-Weyl duality identifies $\mathcal K_\lambda$ with the irreducible representation of $\UU(D)$ labelled by a Young diagram $\lambda\vdash m$ with $\ell(\lambda)\le D$ rows, so the retained dimension is
\begin{equation}\label{eq:KD}
K_D(m):=\sum_{\lambda\vdash m,\ \ell(\lambda)\le D}\dim\mathcal K_\lambda=\bigo{m^{(D-1)(D+2)/2}}.
\end{equation}
\end{lem}
\begin{proof}
Item~1 is Schur's lemma: an operator commuting with every $V_g$ acts as $\bigoplus_\lambda X_\lambda\otimes\id_{\mathcal R_\lambda}$. For item~2, observe that $\mathcal E$ and $\mathcal D$ are compositions of a fixed change of basis, a pinching onto the isotypic blocks, a partial trace and the preparation of a fixed state, hence CPTP. The group twirl annihilates the blocks $\lambda\ne\lambda'$ by Schur orthogonality and acts on each $\mathcal R_\lambda$ as the completely depolarising channel, which is exactly $\mathcal D\circ\mathcal E$. The bound in~\eqref{eq:KD} holds because there are $\bigo{m^{D-1}}$ admissible diagrams and, by the Weyl dimension formula, $\dim\mathcal K_\lambda=\bigo{m^{D(D-1)/2}}$.
\end{proof}

We consider three oracle models for an unknown POVM, differing in the structure made available during storage. The direct model $\Omega_{\mathrm{Dir}}$ provides black-box access to the POVM channel $\PP_{\mathbf M}$ (Definition~\ref{def:direct-oracle}); the structured models $\Omega_{\mathrm{Nai}}$ and $\Omega_{\mathrm{Br}}$ provide coherent access to a fixed Naimark dilation (Definition~\ref{def:coherent-oracle}) and to a Kotowski-Oszmaniec branch decomposition (Definition~\ref{def:branch-oracle}). The representations are given as oracle access, their construction is not addressed, and no reduction between the structured models is known. What unifies the three routes is port-based teleportation, and the lower bounds of Section~\ref{sec:lower} apply to all of them. The universal pPBT protocol applies to every channel, in particular to $\PP_{\mathbf M}$.

\begin{lem}[Universal PSAR for POVM measurement channels]
\label{lem:psar-povm-channel}
For a POVM measurement channel $\mathcal{P}_{\mathbf{M}}: \mathcal{B}(\mathbb{C}^d) \to \mathcal{B}(\mathbb{C}^n)$ with $n$ outcomes and access to $N$ queries, the universal pPBT-based storage-and-retrieval protocol yields:
\begin{itemize}
\item Success probability (state-independent): $\lambda(N,d) = \frac{N}{N-1+d^2}$.
\item Program carrier: dimension $(dn)^N$ before compression.
\item On heralded success: retrieved channel equals original POVM channel exactly.
\end{itemize}
\end{lem}

\begin{proof}
Let $\omega_{AB}=\projj{\zeta}$ be the \emph{optimal} $N$-port pPBT resource in dimension $d$, $\ket{\zeta}=(O_A\otimes\id)\ket{\psi^+}^{\otimes N}$, with $O_A=X_A^{1/2}$, $X_A=\sum_\mu c_\mu P_\mu$, $c_\mu\ge0$, $\Tr X_A=d^N$, and $P_\mu$ the isotypic Young projectors of $S_N$ on $(\C^d)^{\otimes N}$~\cite[Thm.~8]{studzinski2017port} - a single globally optimized entangled state.
Denoting reference systems $A^{N}$ and program systems $B^{N}$, during storage, apply $\PP_{\mathbf M}$ to each program system, obtaining the program
\begin{equation}
    \sigma_{\mathrm{prog}}=\bigl(\mathrm{id}_{A^{N}}\otimes\PP_{\mathbf M}^{\otimes N}\bigr)(\omega_{AB}).
\end{equation}
Thus, each program system becomes an $n$-dimensional classical outcome register, and the program lives on the $N$ reference systems ($\mathbb C^{d}$ each) and the $N$ outcome registers ($\mathbb C^{n}$ each), resulting in a carrier of dimension $(dn)^{N}$. This carrier is not the memory actually required: the program is invariant under relabelling the ports and can be compressed to a size polynomial in $N$ (Proposition~\ref{prop:direct-cq}).

At retrieval, the fresh input $\rho\in\Omega(\mathbb C^{d})$ is teleported by the optimal pPBT protocol: a joint measurement on $\rho$ and the reference systems $A^{N}$ heralds a success flag $\mathsf s$ and selects a port. Because $\PP_{\mathbf M}^{\otimes N}$ acts on the program systems $B^{N}$ while the teleportation measurement acts on $\rho$ and $A^{N}$, the two commute; the result is identical to first teleporting $\rho$ and then applying $\PP_{\mathbf M}$ to the selected port. Conditioned on $\mathsf s$, the outcome register of the selected output port therefore carries $\PP_{\mathbf M}(\rho)$ exactly, via the Choi-Jamio\l{}kowski identity. Due to optimality of $\omega_{AB}$ and given the fact that teleported system has dimension $d$, the success is state-independent with the optimal \cite{sedlak2019optimal,studzinski2017port} pPBT success probability $\lambda(N,d)=\tfrac{N}{N-1+d^{2}}$.
\end{proof}

\section{Storage and retrieval via the POVM channel} \label{direct}

The most direct approach regards the measurement as its quantum-to-classical channel, with black-box access only to $\PP_{\mathbf M}$ through the oracle $\Omega_{\mathrm{Dir}}(\mathbf M)$ of Definition~\ref{def:direct-oracle}. Conditioned on the pPBT success event, the protocol produces a retrieval channel, denoted by $\hat{\PP}^{\mathrm{cond}}_{\mathbf M}$.

\begin{deff}[Direct POVM-channel oracle $\Omega_{\mathrm{Dir}}$]\label{def:direct-oracle}
     A direct oracle $\Omega_{\mathrm{Dir}}(\mathbf M)$ for an $n$-outcome POVM $\mathbf M$ on $\C^{d}$ provides black-box access to the POVM channel $\PP_{\mathbf M}:\BB(\C^{d})\to\BB(\C^{n})$ as a quantum-to-classical channel. A single query allows the user to prepare an arbitrary input state on $\mathbb{C}^d$, including a state entangled with a local ancillary system, and apply $\PP_{\mathbf M}$ to obtain a classical outcome label in $\{1,\ldots,n\}$. This output is stored in the program register.
\end{deff}

Its program is first stated on the uncompressed carrier $(dn)^N$; Section~\ref{sec:direct-compression} shows that it compresses exactly to polynomial size, so the structured constructions of Sections~\ref{naimark} and~\ref{kot-osz} are to be compared with the compressed program rather than with the carrier.

\begin{propo}[Direct POVM-channel construction]\label{prop:direct-ppbt}
Direct application of Lemma~\ref{lem:psar-povm-channel} to the measurement channel $\mathcal{P}_{\mathbf{M}}:\mathcal{B}(\mathbb{C}^{d})\to\mathcal{B}(\mathbb{C}^{n})$ yields:
\begin{itemize}
\item[(i)] \emph{Exactness.}  On the pPBT success flag the retrieved channel equals $\mathcal{P}_{\mathbf{M}}$ exactly, $F_{\mathrm{ch}}\bigl(\hat{\mathcal{P}}^{\mathrm{cond}}_{\mathbf{M}},\mathcal{P}_{\mathbf{M}}\bigr)=1$.
\item[(ii)] \emph{Heralding.}  Success is state-independent with probability
\begin{equation}
    \lambda(N,d)=\frac{N}{N-1+d^{2}}=1-\bigt{\frac{d^{2}}{N}}, \quad N \gg d^2,
\end{equation}
where the success rate depends solely on the input dimension $d$ and the number of oracle uses $N$.
\item[(iii)] \emph{Memory.} Stored uncompressed, the program occupies a carrier of dimension $(dn)^{N}$, i.e.\ $\lceil N\log_2 d\rceil$ qubits and at most $\lceil N\log_2 n\rceil$ classical bits.
\end{itemize}
\end{propo}
\begin{proof}
Direct consequence of Lemma~\ref{lem:psar-povm-channel} applied to $\mathcal{P}_{\mathbf{M}}:\mathcal{B}(\mathbb{C}^d)\rightarrow\mathcal{B}(\mathbb{C}^n)$.
\end{proof}

\subsection{Compressing the direct program}\label{sec:direct-compression}

The program of Proposition~\ref{prop:direct-ppbt} is not stored efficiently on the carrier $(dn)^N$: it is invariant under relabelling the ports and its outcome register is classical, and both properties can be used for compression. For $\pi\in S_N$ let $V^A_\pi$ and $V^C_\pi$ permute the $N$ reference systems and the $N$ outcome registers, and write $\pi\cdot o$ for the permuted word, so that $V^C_\pi\ket{o}=\ket{\pi\cdot o}$ for $o\in[n]^N$.

\begin{lem}[Symmetry of the direct program]\label{lem:direct-symmetry}
For every $n$-outcome POVM $\mathbf M$ on $\C^d$, the program $\sigma_{\mathbf M}=(\mathrm{id}_{A^N}\otimes\PP_{\mathbf M}^{\otimes N})(\omega_{AB})$ of Lemma~\ref{lem:psar-povm-channel} has the form
\begin{equation}\label{eq:direct-cq}
\sigma_{\mathbf M}=\sum_{o\in[n]^N}\rho_{\mathbf M}(o)\otimes\projj{o},
\end{equation}
with positive operators $\rho_{\mathbf M}(o)$ on the reference systems, and it is invariant under joint port permutations, $(V^A_\pi\otimes V^C_\pi)\,\sigma_{\mathbf M}\,(V^A_\pi\otimes V^C_\pi)^\dagger=\sigma_{\mathbf M}$ for all $\pi\in S_N$. Equivalently,
\begin{equation}\label{eq:direct-covariance}
\rho_{\mathbf M}(\pi\cdot o)=V^A_\pi\,\rho_{\mathbf M}(o)\,V^{A\dagger}_\pi\qquad\text{for all }\pi\in S_N,\ o\in[n]^N.
\end{equation}
\end{lem}
\begin{proof}
Equation~\eqref{eq:direct-cq} holds because $\PP_{\mathbf M}$ is a measurement, so every outcome register is classical. For the symmetry, write $\omega_{AB}=\projj{\zeta}$ with $\ket{\zeta}=(O_A\otimes\id)\ket{\psi^+}^{\otimes N}$ and $O_A=X_A^{1/2}$, $X_A=\sum_\mu c_\mu P_\mu$, as in the proof of Lemma~\ref{lem:psar-povm-channel}. Each isotypic Young projector $P_\mu=\tfrac{d_\mu}{N!}\sum_{\sigma\in S_N}\chi^\mu(\sigma^{-1})V^A_\sigma$, where $d_\mu$ and $\chi^\mu$ are the dimension and character of the irreducible representation $\mu$ of $S_N$, is central in the group algebra~\cite[Suppl.\ Inf., Fact~19]{studzinski2017port}; since the $P_\mu$ are mutually orthogonal, $O_A=\sum_\mu\sqrt{c_\mu}\,P_\mu$ commutes with every port permutation, $[O_A,V^A_\pi]=0$. As $\ket{\psi^+}^{\otimes N}$ is invariant under $V^A_\pi\otimes V^B_\pi$, the resource $\omega_{AB}$ is invariant under joint permutations of the ports, and so is $\sigma_{\mathbf M}$, because $\PP_{\mathbf M}^{\otimes N}$ applies the same channel to every port. Inserting~\eqref{eq:direct-cq} into this invariance and comparing the coefficients of $\projj{\pi\cdot o}$ gives~\eqref{eq:direct-covariance}.
\end{proof}

Treating each pair of reference system and outcome register as a single system $\C^{dn}$, Lemmas~\ref{lem:direct-symmetry} and~\ref{lem:inv-compression} already compress $\sigma_{\mathbf M}$ to dimension $K_{dn}(N)=\bigo{N^{(dn-1)(dn+2)/2}}$, polynomial in $N$. Keeping the outcome register classical gives a smaller exponent. For a word $o$ let $t(o)=(t_1,\dots,t_n)$ be its \emph{type}, $t_j=|\{k:o_k=j\}|$; let $o^t$ be the sorted word of type $t$; and fix $\pi_o\in S_N$ with $\pi_o\cdot o^{t(o)}=o$. The entries $t_j\in\{0,1,\dots,N\}$ sum to $N$; zero counts occur, and we set $K_d(0)=1$. Types range over $T_{n,N}=\{t\in\mathbb N^n:\sum_jt_j=N\}$, and the Young subgroup $S_t=S_{t_1}\times\dots\times S_{t_n}$ is the stabiliser of $o^t$.

\begin{propo}[Compressed direct program]\label{prop:direct-cq}
There are encoding and decoding CPTP maps, independent of $\mathbf M$, such that for every $n$-outcome POVM $\mathbf M$ on $\C^d$ the program $\sigma_{\mathbf M}$ is recovered exactly from a hybrid program consisting of a classical type $t\in T_{n,N}$ and a quantum register of dimension $\prod_{j=1}^nK_d(t_j)$. Consequently, the direct protocol of Proposition~\ref{prop:direct-ppbt}, with the same $N$ queries, the same exactness and success probability $\lambda(N,d)$, can be stored with
\begin{equation}\label{eq:direct-cq-memory}
Q=\Bigl\lceil\log_2\max_{t\in T_{n,N}}\prod_{j=1}^nK_d(t_j)\Bigr\rceil,
\qquad
C=\Bigl\lceil\log_2\binom{N+n-1}{n-1}\Bigr\rceil,
\end{equation}
and logical direct-sum dimension $\sum_{t\in T_{n,N}}\prod_jK_d(t_j)=\bigo{N^{(n-1)+n(d-1)(d+2)/2}}$ (Remark~\ref{rem:logical-dim}).
\end{propo}
\begin{proof}
\emph{Encoding.} (1) Measure the outcome registers in the computational basis; by~\eqref{eq:direct-cq} this does not disturb $\sigma_{\mathbf M}$ and yields a word $o$ with probability $\Tr\rho_{\mathbf M}(o)$. (2) Apply $V^{A\dagger}_{\pi_o}$ to the reference systems. By~\eqref{eq:direct-covariance} they are left in the normalised state of $\rho_{\mathbf M}(o^t)$, $t=t(o)$, which depends on $o$ only through $t$; record $t$ and discard $o$. (3) By~\eqref{eq:direct-covariance} with $\pi\in S_t$, the state $\rho_{\mathbf M}(o^t)$ is invariant under $S_t$, which acts on $\bigotimes_j(\C^d)^{\otimes t_j}$ factor by factor. Its isotypic decomposition is therefore the tensor product of those of the $S_{t_j}$, and the map $\mathcal E$ of Lemma~\ref{lem:inv-compression} for $S_t$ retains a register of dimension $\prod_jK_d(t_j)$.

\emph{Decoding.} Given $t$, apply the corresponding map $\mathcal D$, prepare $\ket{o^t}$ on the outcome registers, and apply $V^A_\pi\otimes V^C_\pi$ for a uniformly random $\pi\in S_N$. By Lemma~\ref{lem:inv-compression} the first step returns $\rho_{\mathbf M}(o^t)$ exactly. By~\eqref{eq:direct-covariance} all words of a given type are equally likely and $\rho_{\mathbf M}(\pi\cdot o^t)=V^A_\pi\rho_{\mathbf M}(o^t)V^{A\dagger}_\pi$; since type $t$ is recorded with probability $m_t\Tr\rho_{\mathbf M}(o^t)$, $m_t=N!/\prod_jt_j!$, and each of its $m_t$ words is produced by $N!/m_t$ permutations, the random permutation reproduces the $t$-block of~\eqref{eq:direct-cq}; hence the decoder outputs $\sigma_{\mathbf M}$. Retrieval then proceeds as in Lemma~\ref{lem:psar-povm-channel}, so exactness and $\lambda(N,d)$ are unchanged.

\emph{Memory.} $|T_{n,N}|=\binom{N+n-1}{n-1}=\bigo{N^{n-1}}$ gives $C$, the largest block gives $Q$, and $K_d(t_j)=\bigo{N^{(d-1)(d+2)/2}}$ by~\eqref{eq:KD} gives the total.
\end{proof}

\begin{remark}\label{rem:direct-cq}
Steps~(2) and~(3) of the encoding use different properties of the program. Step~(2) uses~\eqref{eq:direct-covariance}: the port permutation is undone on the reference systems \emph{together with} the outcomes. Storing only the type of the outcome string without this re-alignment discards the pairing between reference systems and outcomes; with the pPBT decoder and any classical post-processing that returns an outcome on every success, retrieval then fails to be exact (Proposition~\ref{prop:type-only}). Additional rejection can restore exactness in special cases at a lower success probability: in the example of Proposition~\ref{prop:type-only}, rejecting the mixed type $(1,1)$ retrieves the computational-basis measurement exactly with success $\tfrac15$. Step~(3) uses that the outcomes are classical, which replaces the symmetry $\UU(dn)$ of the generic compression by $\UU(d)^{\times n}$ acting block by block and lowers the exponent from $(dn-1)(dn+2)/2$ to $(n-1)+n(d-1)(d+2)/2$. For the qubit SIC POVM at $N=50$ the generic compression needs $\lceil\log_2K_8(50)\rceil=61$ qubits, against $(Q,C)=(23,15)$ for~\eqref{eq:direct-cq-memory} (Table~\ref{tab:numerical}). We make no optimality claim for~\eqref{eq:direct-cq-memory}; whether the memory of the direct route can be reduced further, and whether its success probability can be improved under the direct access model, remain open.
\end{remark}

\begin{remark}[Logical dimension and allocations]\label{rem:logical-dim}
Write $D_t=\prod_jK_d(t_j)$ and $T=|T_{n,N}|$. Packing the type blocks into orthogonal quantum sectors, as in~\eqref{eq:direct-rank1-quantum} below, gives a purely quantum program of logical dimension $\sum_tD_t=[z^N](1-z)^{-nd}(1-z^2)^{-nd(d-1)/2}$, by Littlewood's identity $\sum_{m\ge0}K_d(m)z^m=(1-z)^{-d}(1-z^2)^{-d(d-1)/2}$~\cite[Ch.~I.5]{macdonald1995symmetric}. Keeping the type classical instead uses a fixed register of dimension $\max_tD_t$ and an alphabet of size $T$, of unrounded capacity $T\max_tD_t\ge\sum_tD_t$, which~\eqref{eq:direct-cq-memory} rounds separately; the two allocations share the polynomial exponent but not the finite-size total. Since the compressed block of type $t$ is also block diagonal in the Young diagrams $\lambda_j\vdash t_j$ (Lemma~\ref{lem:inv-compression}), these labels may be kept classical as well, retaining only $\bigotimes_j\mathcal K_{\lambda_j}$ quantumly~\cite{yang2016efficient}: then $Q=\lceil\log_2\max_{t,\lambda}\prod_j\dim\mathcal K_{\lambda_j}\rceil$ and $C=\lceil\log_2\sum_t\prod_jp_d(t_j)\rceil$, with $p_d(m)$ the number of partitions of $m$ into at most $d$ parts. For the qubit SIC POVM at $N=50$ the three allocations give $(Q,C)=(36,0)$, $(23,15)$ and $(16,25)$.
\end{remark}

\begin{propo}[Compressed direct program for rank-one POVMs]\label{prop:direct-rank1}
If every effect of $\mathbf M$ has rank at most one, the quantum register of Proposition~\ref{prop:direct-cq} can be replaced, with the same classical type record, the same decoder and the same success probability $\lambda(N,d)$, by a register of dimension $\prod_{j=1}^n\binom{t_j+d-1}{d-1}$, so that
\begin{equation}\label{eq:direct-rank1-memory}
Q_{\mathrm{rank}\text{-}1}=\Bigl\lceil\log_2\max_{t\in T_{n,N}}\prod_{j=1}^n\binom{t_j+d-1}{d-1}\Bigr\rceil,
\end{equation}
and the logical direct-sum dimension $\sum_{t\in T_{n,N}}\prod_{j}\binom{t_j+d-1}{d-1}=\binom{N+nd-1}{N}$ is polynomial of degree $nd-1$ in $N$. Alternatively the type may be stored in the quantum register: embedding the blocks into orthogonal sectors of $\bigoplus_{t\in T_{n,N}}\mathcal S_t$ gives a purely quantum program with the same decoder and success probability and
\begin{equation}\label{eq:direct-rank1-quantum}
Q^{\mathrm q}_{\mathrm{rank}\text{-}1}=\Bigl\lceil\log_2\binom{N+nd-1}{N}\Bigr\rceil,\qquad C=0 .
\end{equation}
\end{propo}
The proof, which identifies the support of the conditional operators of each type as $\bigotimes_j\mathrm{Sym}^{t_j}(\C^d)$, is given in Appendix~\ref{app:rank1-proof}.

\begin{remark}
For the benchmark families of Table~\ref{tab:numerical}, in which all $n$ effects are nonzero and rank one, the purely quantum program~\eqref{eq:direct-rank1-quantum} is never larger or less successful than the minimal Naimark register of Theorem~\ref{thm:ppbt-exact-naimark}: then $D_{\min}=n\ge d$, so $nd\le n^2$, $\binom{N+nd-1}{N}\le\binom{N+n^2-1}{N}$ and $\lambda(N,d)\ge\lambda(N,n)$, with equality in both exactly when $n=d$. If only $k<n$ effects are nonzero, the public rank profile tells both encoders which outcomes are active; then $D_{\min}=k$, the direct encoder uses $k$ outcome types, and the same comparison holds with $n$ replaced by $k$. The refinement uses the promise that the effects are rank one, exactly as the minimal Naimark benchmarks do; no description of the effects themselves is needed. We make no optimality claim for either register.
\end{remark}

The re-alignment in step~(2) cannot be dispensed with if the decoder is to return an outcome on every pPBT success. The reason is a rigidity property of exact port-based teleportation.

\begin{lem}[Rigidity of successful pPBT]\label{lem:ppbt-rigidity}
For the optimal pPBT protocol of Lemma~\ref{lem:psar-povm-channel} and every port $a$,
\begin{equation}\label{eq:ppbt-rigidity}
    \mathcal T_d^{(\textsf{s},a)}(\rho)=\frac{\lambda(N,d)}{N}\,\rho_{B_a}\otimes\xi_a,
\end{equation}
where $\rho_{B_a}$ is $\rho$ placed on port $B_a$ and $\xi_a$ is a state of the remaining $N-1$ ports that does not depend on $\rho$. Each single-port marginal of $\xi_a$ is $\id_d/d$.
\end{lem}
\begin{proof}
Exactness means that tracing out the remaining ports leaves $c\,\rho$ on $B_a$, for a scalar $c$ that may a priori depend on $\rho$; since the map is linear in $\rho$, $c$ is constant, and permutation covariance of the protocol together with~\eqref{eq:ppbt-instrument} gives $c=\lambda(N,d)/N$. Hence $\mathcal T_d^{(\textsf{s},a)}/c$ is a channel whose marginal on $B_a$ is the identity channel. Every Stinespring dilation of the identity channel is the trivial one followed by the preparation of a fixed state of the environment~\cite{watrous2018theory}, which gives~\eqref{eq:ppbt-rigidity}. Finally, $O_A$ is a combination of port permutations, and the optimal POVM of~\cite[Thm.~8]{studzinski2017port} is built from maximally entangled projectors and port permutations, so the protocol is unitarily covariant: $\mathcal T_d^{(\textsf{s},a)}(U\rho U^\dagger)=U^{\otimes N}\,\mathcal T_d^{(\textsf{s},a)}(\rho)\,U^{\otimes N\dagger}$ for every $U\in\UU(d)$. Comparing with~\eqref{eq:ppbt-rigidity}, $\xi_a$ commutes with every $U^{\otimes(N-1)}$, so each of its single-port marginals commutes with every $U$ and equals $\id_d/d$.
\end{proof}

\begin{propo}[Type-only storage of the direct program]\label{prop:type-only}
Suppose the direct program is stored as its reference systems $A^N$ together with only the type $t$ of the outcome word, and that at retrieval the outcome of the selected port is read from a uniformly random word of type $t$. Conditioned on success, the retrieved measurement is then the POVM with effects
\begin{equation}\label{eq:type-only}
    \frac1N\,M_j+\Bigl(1-\frac1N\Bigr)\frac{\Tr M_j}{d}\,\id_d,\qquad j=1,\dots,n,
\end{equation}
that is, $\mathbf M$ preceded by a depolarising channel of visibility $1/N$. More generally, no classical post-processing whose output probabilities depend only on the selected port and on $t$, which returns an outcome on every pPBT success without introducing an additional failure event, restores exactness: for $d=n=N=2$ and the computational-basis measurement, every such post-processing has worst-case total-variation error at least $\tfrac14$.
\end{propo}
\begin{proof}
By Lemma~\ref{lem:ppbt-rigidity} and the commutation argument of Lemma~\ref{lem:psar-povm-channel}, conditioned on success at port $a$ the outcome $o_a$ is distributed according to $\PP_{\mathbf M}(\rho)$ and is independent of the outcomes of the remaining ports, whose distribution does not depend on $\rho$ and has single-port marginals $\Tr(M_j)/d$. A uniformly random word of type $t$ shows outcome $j$ at any fixed position with probability $t_j/N$, and the expectation of $t_j$ conditioned on success is $\Tr(M_j\rho)+(N-1)\Tr(M_j)/d$; dividing by $N$ gives~\eqref{eq:type-only}. For the second claim, let $d=n=N=2$ and $M_1=\projj{0}$, $M_2=\projj{1}$. The remaining port then shows either outcome with probability $\tfrac12$, so for $\rho=\projj{0}$ the type is $(2,0)$ or $(1,1)$ with probability $\tfrac12$ each, and for $\rho=\projj{1}$ it is $(1,1)$ or $(0,2)$. Let $x$ be the probability, averaged over the selected port, that the post-processing outputs $2$ on the type $(1,1)$. The error probability is then at least $x/2$ for $\rho=\projj{0}$ and at least $(1-x)/2$ for $\rho=\projj{1}$, so one of them is at least $\tfrac14$.
\end{proof}

\section{Lower bounds and memory scaling}\label{sec:lower}

The construction of Section~\ref{direct} raises the converse question: how much program memory must \emph{any} universal retrieval processor keep, whatever the access model used during storage? The bounds of this section are stated for arbitrary processors in the sense of Definition~\ref{def:exact-retrieval} and apply to every storage access model, in particular to the structured constructions of Sections~\ref{naimark} and~\ref{kot-osz} below, whenever the processor covers the full target family named in the statement and all target-dependent retrieval information is counted in the program. An elementary program-dimension bound answers the question for the number of outcomes: together with the direct construction of Proposition~\ref{prop:direct-ppbt}, it fixes the dependence of the total memory on $n$. Three further bounds then bring in the input dimension (Propositions~\ref{prop:param-lower} and~\ref{prop:qdim-lower}) and the success probability (Proposition~\ref{prop:success-lower}).

\begin{propo}[Outcome-label lower bound]
\label{prop:label-memory-lower}
Let a fixed retrieval processor exactly implement every $n$-outcome POVM on $\mathbb C^d$, with positive heralded success probability for every input and every target POVM, and suppose that all measurement-dependent information available at retrieval is contained in a program register of dimension $K$. Then $K\ge n$, so a program stored in $Q$ qubits and $C$ classical bits satisfies $Q+C\ge\log_2 n$.
\end{propo}

\begin{proof}
Consider the $n$ POVMs $\mathbf M^{(j)}$ with $M_i^{(j)}=\delta_{ij}\,\id_d$. Each one is a valid POVM, in the sense that its non-selected effects vanish, which is admissible since effects need only be positive semidefinite with $\sum_i M^{(j)}_i=\id_d$. Assume that each $\mathbf M^{(j)}$ returns outcome $j$ deterministically. Let $\sigma_j$ be the corresponding program states and fix any input state. The processor induces a POVM $\{F_1,\dots,F_n,F_{\rm fail}\}$ on the program register, with $F_j$ the effect for a heralded output of label $j$. 

Exactness gives $\operatorname{Tr}(F_k\sigma_j)=0$ for $k\ne j$, and positive success gives $\operatorname{Tr}(F_j\sigma_j)>0$. Since $F_k$ and $\sigma_j$ are positive semidefinite, the first condition means $F_k\sigma_j=0$: the support of $\sigma_j$ lies in the kernel of every $F_k$ with $k\ne j$. By the second condition we can choose $v_j$ in the support of $\sigma_j$ with $F_jv_j\ne0$. If $\sum_ja_jv_j=0$, applying $F_k$ gives $a_kF_kv_k=0$, hence $a_k=0$; the $n$ vectors $v_j$ are therefore linearly independent, and $K\ge n$. A hybrid program of $Q$ qubits and $C$ classical bits embeds in a register of dimension at most $2^{Q+C}$ (the classical part as a diagonal register), so $Q+C\ge\log_2 K\ge\log_2 n$.
\end{proof}

\begin{remark}\label{rem:label-lower-sideinfo}
The hypothesis that all measurement-dependent information is in the program excludes side information whose form depends on the target. A minimal Naimark readout is such information when the rank profile varies across the family: for the deterministic-label POVMs $\mathbf M^{(j)}$ of the proof, $D_{\min}=d$ and the minimal readout is $\Pi^{(0)}_j=\id_d$ with all other projectors zero, so supplying it free reveals $j$ and would bypass the bound. The Naimark benchmarks of Section~\ref{comparison} use a fixed rank-one profile, for which the hypothesis holds.
\end{remark}

Under direct access the bound is attained, up to constants, for every target success below one.

\begin{cor}[Outcome-count scaling under direct access]\label{cor:direct-outcome-scaling}
Fix $d\ge2$ and a target success $0<p_0<1$, and let $N=\bigl\lceil p_0(d^2-1)/(1-p_0)\bigr\rceil$, which is independent of $n$. Under direct access, every $n$-outcome POVM admits exact heralded retrieval from $N$ queries with success probability at least $p_0$ and
\begin{equation}
 Q=\lceil N\log_2 d\rceil=O_{d,p_0}(1),\qquad C\le\lceil N\log_2 n\rceil=O_{d,p_0}(\log n),
\end{equation}
and the minimum total persistent program memory $Q+C$ is $\Theta_{d,p_0}(\log n)$ as $n\to\infty$.
\end{cor}
\begin{proof}
The choice of $N$ gives $\lambda(N,d)=N/(N-1+d^2)\ge p_0$, and Proposition~\ref{prop:direct-ppbt} yields the stated $Q$ and $C$; the compressed program of Proposition~\ref{prop:direct-cq} needs at most as much. Proposition~\ref{prop:label-memory-lower} gives $Q+C\ge\log_2 n$.
\end{proof}

Under branch-oracle access (Section~\ref{kot-osz} below) the same scaling holds for target success below $\tfrac18$, with a classical record that does not grow with the number of queries.

\begin{cor}[Tight outcome-count scaling at fixed input dimension]
\label{cor:outcome-scaling}
Fix $d\ge2$ and a target success $0<p_0<\tfrac18$. Under branch-oracle access during storage (Definition~\ref{def:branch-oracle} below), with no measurement-dependent information supplied free at retrieval, every $n$-outcome POVM admits exact heralded retrieval with
\begin{equation}
 Q=O_{d,p_0}(1),\qquad C=O_d(\log(n+1)),
\end{equation}
and the minimum total persistent program memory $Q+C$ is $\Theta_{d,p_0}(\log n)$ as $n\to\infty$.
\end{cor}

\begin{proof}
Take the universal value $q=\tfrac18$ of Definition~\ref{def:branch-oracle} and $N=\bigl\lceil p_0(4d^2-1)/(q-p_0)\bigr\rceil$, which is independent of $n$ and satisfies $qN/(N+4d^2-1)\ge p_0$. Proposition~\ref{prop:finite-readout} then yields the stated $Q$ and $C$; Proposition~\ref{prop:label-memory-lower} gives $Q+C\ge\log_2 n$.
\end{proof}

\begin{remark}\label{rem:direct-outcome-scaling}
The $\Theta(\log n)$ scaling is thus not specific to any representation: under direct access it holds for every target success below one, under branch access (Corollary~\ref{cor:outcome-scaling}) for $p_0<\tfrac18$. For $d=2$ a single query succeeds with probability $\lambda(1,2)=\tfrac14$, so for every $p_0\le\tfrac14$, in particular the whole range of Corollary~\ref{cor:outcome-scaling}, the direct route needs one qubit and $\lceil\log_2 n\rceil$ bits, less than the branch construction in both registers. For general $d$, a target $p_0<\tfrac18$ requires at most $\lceil(d^2-1)/7\rceil$ direct queries. The branch construction differs in that its classical record, $2d\lceil\log_2(n+1)\rceil$ bits, does not grow with the number of queries, whereas the direct record does: linearly in $N$, i.e.\ as $1/(1-p_0)$ when $p_0\to1$, if stored uncompressed, and as $(n-1)\log_2N+O_n(1)$ if compressed as in the Proposition~\ref{prop:direct-cq}. 
\end{remark}

This separates the cost of the outcome alphabet from that of the coherent program: outcome-independent \emph{quantum} memory is compatible with exact heralded retrieval, whereas outcome-independent \emph{total} program memory is not, at fixed positive success.
The bound is driven by rank-deficient test POVMs with vanishing effects, which Definition~\ref{def:povm} admits. It concerns the whole class of $n$-outcome POVMs and a universal processor, so it is not a lower bound for a single fixed measurement such as the qubit SIC, nor for a processor that receives measurement-dependent side information for free. It fixes the dependence on $n$ only at fixed $d$ and $p_0$ (by Propositions~\ref{prop:param-lower} and~\ref{prop:success-lower} below), although it does not establish optimality of the PSAR quantum register, of the query complexity, or of the constant multiplying $\log n$ in the classical upper bound.

The label bound counts distinguishable targets, as do the dimension bounds known for programmable processors: distinct unitaries require mutually orthogonal programs in a deterministic gate array~\cite{nielsen1997programmable}, and the program dimension of approximate programmable measurements grows with the size of a well-separated family of target measurements~\cite{perezgarcia2006optimality}. It treats a qubit like a classical bit and does not involve the input dimension. Counting parameters instead of labels gives a bound that sees both.

Both parameter-counting arguments below use that exactness already forces an input-independent success probability, the property that postselected simulation of measurements~\cite{oszmaniec2018simulating} builds in by construction. The nonscalar hypothesis of the lemma cannot be dropped: for $M_i=a_i\id_d$ the effects $A_i=a_iS$ with any nonscalar $0<S\le\id_d$ reproduce the conditional statistics with input-dependent success $\Tr[S\rho]$.

\begin{lem}[Input-independent success]\label{lem:const-success}
Let a fixed processor retrieve $\mathbf M$ exactly with heralded success in the sense of~\eqref{eq:exact-retrieval}, let $F_i$ be its effect on the program register and the input for a heralded output $i\le n$, and let $A_i:=\Tr_{\mathrm{prog}}\bigl[F_i(\sigma_{\mathbf M}\otimes\id_d)\bigr]$, so that $\Pr[\text{output }i\mid\rho]=\Tr[A_i\rho]$, with $S:=\sum_{i\le n}A_i$. If some effect $M_i$ is not proportional to $\id_d$, then $S=s\,\id_d$ with $s>0$, i.e.\ the success probability $s_{\mathbf M}(\rho)=\Tr[S\rho]=s$ does not depend on the input, and $A_i=s\,M_i$ for every $i$.
\end{lem}
\begin{proof}
Exactness reads $\Tr[A_i\rho]=\Tr[S\rho]\,\Tr[M_i\rho]$ for every state $\rho$. Both sides are polynomials in the entries of $\rho$ and agree on the state space, which has nonempty interior in the real affine hyperplane $\{\Tr\rho=1\}$ of Hermitian matrices; hence they agree on the whole hyperplane. Write $\rho=\id_d/d+Y$ with $Y$ traceless Hermitian. The left-hand side is affine in $Y$, so the terms quadratic in $Y$ on the right must vanish: $\Tr[SY]\,\Tr[M_iY]=0$ for every traceless Hermitian $Y$. A product of two linear functionals on a real vector space vanishes identically only if one of them does, and $Y\mapsto\Tr[M_iY]$ does not, since $Y=M_i-\Tr(M_i)\id_d/d\neq0$ gives $\Tr[M_iY]=\Tr[Y^2]>0$. Hence $\Tr[SY]=0$ for every traceless $Y$, so $S=s\,\id_d$ with $s=\Tr S/d>0$. Then $\Tr[(A_i-sM_i)\rho]=0$ for every trace-one Hermitian $\rho$, which forces $A_i-sM_i=0$: the constant term gives $\Tr(A_i-sM_i)=0$ and the linear terms give $\Tr[(A_i-sM_i)Y]=0$ for all traceless $Y$.
\end{proof}

\begin{propo}[Parameter-counting lower bound]\label{prop:param-lower}
Under the hypotheses of Proposition~\ref{prop:label-memory-lower}, let $d\ge2$ and let the program be stored in $Q$ qubits and $C$ classical bits. Then
\begin{equation}\label{eq:param-lower}
    2Q+C\;\ge\;\log_2\bigl((n-1)d^2+1\bigr).
\end{equation}
Both this bound and Proposition~\ref{prop:label-memory-lower} hold for programs prepared at random, with exactness understood after averaging the sub-normalized outcome probabilities over the preparation, and neither requires the preparation to depend continuously on $\mathbf M$.
\end{propo}
\begin{proof}
For $n=1$ the only POVM is $(\id_d)$ and~\eqref{eq:param-lower} reads $2Q+C\ge0$; let $n\ge2$. Let $\sigma_{\mathbf M}$ be the program averaged over its preparation. It is block diagonal, with $2^C$ blocks of size $2^Q$, and such operators form a real vector space of dimension $2^C4^Q$. Let $F_i$ be the effect of the processor, acting on the program register and the input, for a heralded output $i\le n$. Outcome $i$ then occurs with probability $\Tr[F_i(\sigma_{\mathbf M}\otimes\rho)]=\Tr[\rho\,\Gamma_i(\sigma_{\mathbf M})]$, where $\Gamma_i(\sigma)=\Tr_{\mathrm{prog}}[F_i(\sigma\otimes\id_d)]$ is linear in $\sigma$; write $\Gamma=\sum_{i\le n}\Gamma_i$. Exactness states that $\Tr[\rho\,\Gamma_i(\sigma_{\mathbf M})]=\Tr[\rho\,\Gamma(\sigma_{\mathbf M})]\,\Tr(M_i\rho)$ for every $\rho$. If some $M_i$ is not proportional to $\id_d$, Lemma~\ref{lem:const-success} (with $A_i=\Gamma_i(\sigma_{\mathbf M})$) gives $\Gamma(\sigma_{\mathbf M})=p\,\id_d$ with $p>0$ and $\Gamma_i(\sigma_{\mathbf M})=p\,M_i$ for all $i$, so the image of the linear map $\sigma\mapsto(\Gamma_i(\sigma))_{i\le n}$ contains a positive multiple of every POVM with all $M_i\succ0$ and some $M_i$ not proportional to $\id_d$. These POVMs form a non-empty open subset of the affine space $\{\sum_iM_i=\id_d\}$ of real dimension $(n-1)d^2$, which does not contain $0$, so their span has dimension $(n-1)d^2+1$. Since the rank of a linear map is at most the dimension of its domain, $(n-1)d^2+1\le2^C4^Q$.
\end{proof}

\begin{remark}\label{rem:param-lower}
The bound~\eqref{eq:param-lower} holds for every success probability and every access model, and it is nearly attained at low success. The direct protocol with a single query, which is the teleportation-based probabilistic gate array of~\cite{nielsen1997programmable} applied to the POVM channel, stores one $d$-dimensional reference system and one outcome and succeeds with probability $\lambda(1,d)=1/d^2$; the dimension count of the proof evaluates to $nd^2$ for its program, against $(n-1)d^2+1$. In this count a quantum register of dimension $D$ carries $D^2$ real parameters and a classical one only $D$, so a qubit weighs as much as two classical bits (a property of the parameter count, not an achievable exchange rate between the two resources: Proposition~\ref{prop:qdim-lower} shows that no finite classical budget removes the need for quantum memory), and keeping the outcome record classical, as in Proposition~\ref{prop:direct-cq}, is the cheaper choice: storing the outcomes in qubits, as the generic compression of Section~\ref{sec:direct-compression} does, counts twice as much per stored bit. Together with Proposition~\ref{prop:label-memory-lower}, every admissible program satisfies both $Q+C\ge\log_2 n$ and~\eqref{eq:param-lower}; the first bound is the stronger one when $C<2\log_2 n-\log_2\bigl((n-1)d^2+1\bigr)$, roughly $C<\log_2(n/d^2)$, and the second otherwise. The bound does not depend on the success probability, whereas all constructions in this paper need more memory as their success probability approaches its supremum. Proposition~\ref{prop:success-lower} below shows that such growth is unavoidable as the success probability approaches one.
\end{remark}

\begin{propo}[Quantum memory cannot be eliminated]\label{prop:qdim-lower}
Under the hypotheses of Proposition~\ref{prop:label-memory-lower}, let $d,n\ge2$ and let the program consist of a label from a finite classical alphabet and a quantum register of dimension $r$. Then $r\ge d$; in particular $Q\ge\lceil\log_2d\rceil$ for every finite classical budget. Moreover, at $r=d$ the counting of Proposition~\ref{prop:param-lower} forces a classical alphabet of size $L\ge n$, and the single-query direct protocol attains logical dimensions $(r,L)=(d,n)$ at success $1/d^2$.
\end{propo}
\begin{proof}
\emph{Reduction and induced effects.} Restrict the target family to $\mathbf M_\psi=(\projj\psi,\ \id_d-\projj\psi,\ 0,\dots,0)$, $\psi$ a unit vector. The program for $\psi$ is a classical label $c\in[L]$ together with $\tau_{\psi,c}\ge0$ on $\C^r$, $\sum_c\Tr\tau_{\psi,c}=1$, randomness of the preparation absorbed. For label $c$ the processor applies effects $F^c_1,F^c_2,F^c_\emptyset\ge0$ on $\C^r\otimes\C^d$ summing to the identity. Set $A_\psi:=\sum_c\Tr_r\bigl[F^c_1(\tau_{\psi,c}\otimes\id_d)\bigr]$ and $B_\psi$ likewise with $F^c_2$; both are positive and linear in $\rho$ in the sense that $\Pr[\text{output }1\mid\rho]=\Tr[A_\psi\rho]$ and $\Pr[\text{output }2\mid\rho]=\Tr[B_\psi\rho]$.

\emph{Success is constant.} For $d\ge2$ the effect $\projj\psi$ is not scalar, so Lemma~\ref{lem:const-success}, applied to the block-diagonal program $\sum_c\projj c\otimes\tau_{\psi,c}$, gives $A_\psi+B_\psi=s_\psi\id_d$ with $s_\psi>0$ and $A_\psi=s_\psi\projj\psi$.

\emph{Decomposition.} Write the spectral decomposition $F^c_1=\sum_a\projj{w_{ca}}$ with eigenvalues absorbed, at most $rd$ terms, and expand $\ket{w_{ca}}=\sum_{m=1}^r\ket m\otimes\ket{w_{ca,m}}$ in a basis of $\C^r$. Let $S_{ca}:=\operatorname{span}\{\ket{w_{ca,m}}\}_{m\le r}\subseteq\C^d$; then $\dim S_{ca}\le r$, and $S_{ca}$ depends only on the processor, not on $\psi$. Write $\tau_{\psi,c}=\sum_b\projj{v_{\psi,cb}}$ with at most $r$ terms. A direct computation gives $\Tr_r\bigl[\projj{w}(\projj v\otimes\id_d)\bigr]=\projj u$ with $\ket u=(\bra v\otimes\id_d)\ket w=\sum_m\braket v m\ket{w_{ca,m}}\in S_{ca}$. Hence
\[
A_\psi=\sum_{c,a,b}\projj{u_{\psi,cab}},\qquad \ket{u_{\psi,cab}}\in S_{ca}.
\]

\emph{Support.} Since $A_\psi=s_\psi\projj\psi$ with $s_\psi>0$: for every $\phi\perp\psi$, $0=\bra\phi A_\psi\ket\phi=\sum|\braket\phi{u_{\psi,cab}}|^2$, so every $u_{\psi,cab}\in\operatorname{span}\{\psi\}$, and $A_\psi\neq0$ forces some $u_{\psi,cab}\neq0$. Thus $\psi\in S_{ca}$ for some pair $(c,a)$.

\emph{Counting.} The subspaces $S_{ca}$, $(c,a)\in[L]\times[rd]$, are finitely many and independent of $\psi$, and every unit vector of $\C^d$ lies in their union. If $r<d$, each $S_{ca}$ is a proper subspace, and a finite union of proper subspaces of $\C^d$ is a proper subset (a finite union of proper subspaces of a complex vector space is a proper subset and, being closed under scaling, misses some unit vector). This contradicts the previous step, so $r\ge d$.

\emph{Second statement.} The proof of Proposition~\ref{prop:param-lower} counts real parameters of block-diagonal programs; with $L$ blocks of size $r$ it reads $Lr^2\ge(n-1)d^2+1$. At $r=d$ this is $L\ge n-1+d^{-2}$, i.e.\ $L\ge n$ for integer $L$. The direct protocol with a single query stores one $d$-dimensional reference system and one of $n$ outcomes, and succeeds with probability $\lambda(1,d)=1/d^2$.
\end{proof}

\begin{remark}
The bound $r\ge d$ uses positivity of the induced effects and holds for mixed program blocks and arbitrary fixed retrieval operations; it concerns a finite classical alphabet, and an exact real-valued classical description is outside its scope.
\end{remark}

A dependence on the success probability follows by reduction to approximate programming; the contribution here is the reduction, while the dimension bound it invokes is due to Pérez-García. Pérez-García~\cite[Thm.~1 and Sec.~III]{perezgarcia2006optimality} considers fixed devices that, for every $d$-outcome POVM $\mathbf P$ on $\C^d$, admit a program state whose programmed effects $G_j$ satisfy $\max_\rho\sum_j|\Tr[\rho(P_j-G_j)]|\le\delta$. He shows that their program space has dimension at least $k_d\,\delta^{-(d-1)/2}$ whenever $\delta<1/(800d)$, where $1/k_d=2d\,(20\sqrt{8d})^{d-1}$, and at least $1/(12800\,\delta)$ for $d=2$ and $\delta<1/6400$.

\begin{propo}[Success-dependent lower bound]\label{prop:success-lower}
Under the hypotheses of Proposition~\ref{prop:label-memory-lower}, let $n\ge d\ge2$ and let the heralded success probability be at least $p$ for every input and every target POVM. If $1-p<1/(1600d)$, then
\begin{equation}\label{eq:success-lower}
    K\;\ge\;k_d\,\bigl(2(1-p)\bigr)^{-(d-1)/2},
\end{equation}
and for $d=2$ and $1-p<1/12800$, $K\ge1/\bigl(25600\,(1-p)\bigr)$. Consequently 
\[Q+C\ge\frac{d-1}{2}\log_2\frac{1}{1-p}-O_d(1) \text{ as } p\to1,\]
and $Q+C\ge\log_2\frac{1}{1-p}-\log_2 25600$ for $d=2$.
\end{propo}
\begin{proof}
On failure, let the processor output label $j$ with a fixed probability $q_j$. This yields a fixed deterministic device whose effects on the input are linear in the program state, as in the proof of Proposition~\ref{prop:param-lower}. For the program of $\mathbf M$, and with success probability $s(\rho)\ge p$ on input $\rho$, its outcome probabilities are $s(\rho)\Tr(M_j\rho)+(1-s(\rho))q_j$, so that
\begin{equation}
    \sum_j\bigl|\Tr(M_j\rho)-s(\rho)\Tr(M_j\rho)-(1-s(\rho))q_j\bigr|=(1-s(\rho))\sum_j\bigl|\Tr(M_j\rho)-q_j\bigr|\le2(1-p)
\end{equation}
for every $\rho$. A $d$-outcome POVM is an $n$-outcome POVM whose effects $M_j$ vanish for $j>d$, and merging the outputs $j\ge d$ into a single one does not increase the left-hand side. The resulting $d$-outcome device therefore approximates every $d$-outcome POVM with $\delta=2(1-p)$; its program space is the program register of dimension $K$, and the bound quoted above gives~\eqref{eq:success-lower} and its version for $d=2$. Finally, $Q+C\ge\log_2K$.
\end{proof}

\begin{cor}[Success-dependent bound for all $d,n\ge2$]\label{cor:success-all}
Under the hypotheses of Proposition~\ref{prop:label-memory-lower}, let $d,n\ge2$ and let the heralded success probability be at least $p$ for every input and every target POVM, with $1-p<1/12800$. Then $K\ge1/\bigl(25600\,(1-p)\bigr)$ and $Q+C\ge\log_2\frac1{1-p}-\log_2 25600$.
\end{cor}
\begin{proof}
Fix an isometry $V:\C^2\to\C^d$. For a binary qubit POVM $(P_1,P_2)$ let $M_1=VP_1V^\dagger$, $M_2=VP_2V^\dagger+\id_d-VV^\dagger$ and $M_j=0$ for $j\ge3$, an $n$-outcome POVM on $\C^d$ with $\Tr(M_iV\rho V^\dagger)=\Tr(P_i\rho)$. Feeding the processor the input $V\rho V^\dagger$ and merging every heralded label other than $1$ into $2$ gives a fixed processor that retrieves every binary qubit POVM exactly with success at least $p$ from the same program register of dimension $K$; the case $d=2$ of Proposition~\ref{prop:success-lower} applies.
\end{proof}

\begin{remark}\label{rem:success-lower}
At fixed $d,n\ge2$ this matches the growth of the compressed direct program up to the constant in front of the logarithm. Success $p$ requires $N=\lceil p(d^2-1)/(1-p)\rceil$ queries, and by Proposition~\ref{prop:direct-cq} the program then needs 
\[
Q+C\le\bigl((n-1)+n(d-1)(d+2)/2\bigr)\log_2\frac1{1-p}+O_{n,d}(1).
\] The total memory needed for exact heralded retrieval with success $p\to1$ is therefore $\Theta_{n,d}\bigl(\log\frac1{1-p}\bigr)$. The constants inherited from~\cite{perezgarcia2006optimality} are far from optimal, and~\eqref{eq:success-lower} is informative only very close to $p=1$: for $d=2$ it guarantees more than ten bits only when $1-p<4\cdot10^{-8}$.
\end{remark}

\section{Storage and retrieval via Naimark dilation} \label{naimark}

The direct construction of Section~\ref{direct} uses only the permutation symmetry of its program and the classical nature of the outcomes, and the bounds of Section~\ref{sec:lower} apply to every access model. The next two sections apply both ingredients to structured oracles: the full-unitary constructions obtained by importing unitary PSAR serve as reference implementations, and the measurement-specific encoder improves on them (Proposition~\ref{prop:direct-rank1}, Corollary~\ref{cor:measured-branch-quantum}). We first consider a richer oracle model based on a Naimark dilation. Since the dilation is unitary, the optimal PSAR protocol for unitary channels can be applied directly. The program is then purely quantum, with register $\binom{N+D_{\mathrm{Nai}}^2-1}{N}$, but its success probability is governed by the dilation dimension $D_{\mathrm{Nai}}\ge d$ rather than by $d$. 

Throughout the remainder of the paper, we denote by $\HH_M$ the program register used to store the representation of the unknown measurement.

\begin{deff}[Coherent (Naimark) measurement oracle $\Omega_{\mathrm{Nai}}$ ]\label{def:coherent-oracle}
A \emph{coherent Naimark oracle} $\Omega_{\mathrm{Nai}}(\mathbf{M})$ for an $n$-outcome POVM $\mathbf{M}$ on $\C^{d}$ consists of a black-box Naimark unitary $U_{\mathbf{M}}\in\UU(\C^{D_{\mathrm{Nai}}})$ on the chosen dilation space of dimension $D_{\mathrm{Nai}}$, a fixed publicly known projective partition $\{\Pi^{(0)}_{i}\}$ of $\id_{D_{\mathrm{Nai}}}$, and a fixed embedding isometry $\eta_{\mathrm{Nai}}:\C^{d}\hookrightarrow\C^{D_{\mathrm{Nai}}}$, such that 
\begin{equation}
    M_{i}=\eta_{\mathrm{Nai}}^{\dagger}U_{\mathbf{M}}^{\dagger}\Pi^{(0)}_{i}U_{\mathbf{M}}\eta_{\mathrm{Nai}}, \quad \text{for every} \quad i.
\end{equation}

A query consumes one coherent use of $U_{\mathbf{M}}$ on an arbitrary input state of $\C^{D_{\mathrm{Nai}}}$, possibly entangled with a reference system, and not only on embedded states $\eta_{\mathrm{Nai}}\rho\,\eta_{\mathrm{Nai}}^{\dagger}$; the readout $\{\Pi^{(0)}_{i}\}$ is performed by the experimenter. The dilation space, the embedding $\eta_{\mathrm{Nai}}$ and the readout are fixed across the family of target POVMs under consideration and are the only side information available at retrieval; they carry no information about the particular $\mathbf M$ beyond this public promise. In particular, a minimal realisation $D_{\mathrm{Nai}}=D_{\min}$ forces $\operatorname{rank}\Pi^{(0)}_i=\operatorname{rank}M_i$, so minimal-dilation benchmarks are comparisons within a publicly promised rank profile, such as the rank-one profile used in Section~\ref{comparison}; for a family with varying rank profiles one must use a readout with a target-independent profile, e.g.\ the canonical one, or account for the information carried by the profile (Remark~\ref{rem:label-lower-sideinfo}).
\end{deff}

Recall from Theorem~\ref{thm:naimark} that every POVM $\mathbf M$ admits an isometric dilation $V:\mathbb C^d\to\mathbb C^{D_{\mathrm{Nai}}}$ to a projective measurement $\{\Pi^{(0)}_i\}$ on $\C^{D_{\mathrm{Nai}}}$, such that 
\begin{equation} \label{eq:naimark-dilation}
    M_i=V^\dagger \Pi^{(0)}_i V.
\end{equation}
We fix a unitary extension $U_{\mathbf M}\in\mathcal U(D_{\mathrm{Nai}})$ satisfying $V=U_{\mathbf M}\eta_{\mathrm{Nai}}$ and identify this unitary with the unknown unitary provided by the coherent Naimark oracle.

Since the Naimark oracle provides coherent access to the unknown unitary, the storage-and-retrieval protocol is the unitary PSAR protocol followed by the fixed readout.

\begin{deff}[The protocol $\LL_{\mathrm{Nai,pPBT}}$ for Naimark route] \label{protocolNaimark}
Let $U_{\mathbf M}\in\UU(D_{\mathrm{Nai}})$ be the Naimark unitary provided by the oracle $\Omega_{\mathrm{Nai}}(\mathbf{M})$ of Definition~\ref{def:coherent-oracle}, with associated isometry
\begin{equation}
    V:=U_{\mathbf M} \eta_{\mathrm{Nai}} \in\BB(\C^{d},\C^{D_{\mathrm{Nai}}}),
\end{equation}
satisfying the Naimark dilation relation~\eqref{eq:naimark-dilation}.

\emph{\textbf{Storage.}} Run the optimal $N\to1$ PSAR of~\cite{sedlak2019optimal} on the Naimark unitary, producing the optimal PSAR program state whose dimension is given by~\eqref{eq:psar-memory}.

\emph{\textbf{Retrieval.}} For $\rho\in\Omega(\C^{d})$, form the embedded state $\eta_{\mathrm{Nai}}\rho\eta_{\mathrm{Nai}}^{\dagger}\in\Omega(\C^{D_{\mathrm{Nai}}})$ and run the pPBT retrieving instrument.  On failure flag $\textsf{f}$ abort; on success $\textsf{s}$ the output port carries $U_{\mathbf M}(\eta_{\mathrm{Nai}}\rho\eta_{\mathrm{Nai}}^{\dagger})U_{\mathbf M}^{\dagger}\in\Omega(\C^{D_{\mathrm{Nai}}})$ exactly, to which the fixed readout $\{\Pi^{(0)}_{i}\}$ is applied.

The resulting conditional retrieval channel is denoted by $\hat{\PP}^{\mathrm{cond}}_{\mathbf{M}}$. 
\end{deff}

\begin{thm}[Exact heralded SAR of POVMs]\label{thm:ppbt-exact-naimark}
The protocol achieves exact retrieval conditioned on success while retaining the memory scaling of the optimal-success PSAR register. Let $\mathbf{M}$ be an arbitrary $n$-outcome POVM on $\C^{d}$, realised as a projective measurement on $\C^{D_{\mathrm{Nai}}}$, where $D_{\mathrm{Nai}}$ is the dimension of the chosen Naimark dilation. Let $\LL_{\mathrm{Nai,pPBT}}$ store the corresponding simulation unitary through the optimal $N$-port PSAR of~\cite{sedlak2019optimal}. Then: 
\begin{itemize}
\item[(i)] \emph{Exactness.}  Conditioned on the heralded success event the retrieved measurement reproduces $\mathbf{M}$ \emph{exactly}, $F_{\mathrm{ch}}\bigl(\hat{\PP}^{\mathrm{cond}}_{\mathbf{M}},\PP_{\mathbf{M}}\bigr)=1$.
\item[(ii)] \emph{Heralding.}  Success is flagged and state-independent. For the Naimark realisation we have 
\begin{equation}
\Pr[\textsf{s}]=\lambda(N,D_{\mathrm{Nai}}),
\end{equation}
\item[(iii)] \emph{Memory.}  The program register of the optimal-success PSAR protocol has dimension $\binom{N+D_{\mathrm{Nai}}^{2}-1}{N}=\bigt{N^{D_{\mathrm{Nai}}^{2}-1}}$, polynomial in $N$ at fixed $D_{\mathrm{Nai}}$.
\end{itemize}
\end{thm}

\begin{proof}
For part (i) we rely on \eqref{eq:ppbt-instrument}. Conditioned on the success flag $\textsf{s}$, the pPBT retrieval implements the identity channel on $\C^{D_{\mathrm{Nai}}}$, so the output equals $U_{\mathbf M}(\eta_{\mathrm{Nai}}\rho\eta_{\mathrm{Nai}}^{\dagger})U_{\mathbf M}^{\dagger}$. Applying the fixed readout and the dilation identity \eqref{eq:naimark-dilation} yields $\Pr[i\mid\textsf{s}]=\mathrm{Tr}(M_i\rho)$, i.e.\ $F_{\mathrm{ch}}(\hat{\PP}^{\mathrm{cond}}_{\mathbf M},\PP_{\mathbf M})=1$. For the second part, the success is the state-independent pPBT probability $\lambda(N,D_{\mathrm{Nai}})$, depending on $\mathbf M$ only through the dilation dimension. Last point follows from  Theorem~\ref{PSAR}, which gives
\begin{equation}
    \dim\HH_{M}=\binom{N+D_{\mathrm{Nai}}^{2}-1}{N}= \bigt{N^{D_{\mathrm{Nai}}^2-1}}.
\end{equation}
The canonical dilation has $D_{\mathrm{Nai}}=dn$, and the minimal one is discussed in Remark \ref{rem:naimark-dim}.
\end{proof}

\begin{remark}[Canonical versus minimal dilation dimension]\label{rem:naimark-dim}
The construction presented in Theorem~\ref{thm:ppbt-exact-naimark} depends on the dilation through its dimension rather than other characteristics: for any Naimark realisation of $\mathbf{M}$ as a projective measurement on $\C^{D_{\mathrm{Nai}}}$, the heralded success probability is $\lambda(N,D_{\mathrm{Nai}})$ and the program register has dimension $\binom{N+D_{\mathrm{Nai}}^{2}-1}{N}$.
We work with both canonical ($D_{\mathrm{Nai}}=dn$) and minimal ($D_{\mathrm{Nai}}=D_{\min}$, see \eqref{eq:Dmin}) realisations throughout this paper.
For example, the qubit SIC POVM ($d=2$, $n=4$) then admits $D_{\min}=4$ rather than requiring canonical $8$-dimensional dilation. 
\end{remark}

\begin{propo}[Measured Naimark program]\label{prop:measured-naimark}
Refine the readout of Definition~\ref{def:coherent-oracle} publicly into rank-one projectors, $\Pi^{(0)}_i=\sum_{l:g(l)=i}\projj l$ for a fixed orthonormal basis $\{\ket l\}_{l\le D_{\mathrm{Nai}}}$ and a fixed map $g$. Then $R_l:=\eta_{\mathrm{Nai}}^\dagger U_{\mathbf M}^\dagger\projj lU_{\mathbf M}\eta_{\mathrm{Nai}}$ is a $D_{\mathrm{Nai}}$-outcome POVM on $\C^d$ with effects of rank at most one and $\sum_{l:g(l)=i}R_l=M_i$. Applying each of the $N$ coherent calls to an embedded $d$-dimensional port of the optimal pPBT resource, followed by this readout, encoding the $N$ uses of $\mathbf R$ with Proposition~\ref{prop:direct-rank1} with the type in the quantum register, and relabelling the retrieved outcome by $g$, gives exact heralded retrieval of every target $\mathbf M$, whatever the ranks of its effects, with
\begin{equation}\label{eq:measured-naimark}
Q=\Bigl\lceil\log_2\binom{N+dD_{\mathrm{Nai}}-1}{N}\Bigr\rceil,\qquad C=0,\qquad \Pr[\mathsf s]=\lambda(N,d).
\end{equation}
Since $dD_{\mathrm{Nai}}\le D_{\mathrm{Nai}}^2$ and $\lambda(N,d)\ge\lambda(N,D_{\mathrm{Nai}})$, both strictly when $D_{\mathrm{Nai}}>d$, this program is never larger or less successful than that of Theorem~\ref{thm:ppbt-exact-naimark}.
\end{propo}
\begin{proof}
$\sum_lR_l=\eta_{\mathrm{Nai}}^\dagger U_{\mathbf M}^\dagger U_{\mathbf M}\eta_{\mathrm{Nai}}=\id_d$, $\operatorname{rank}R_l\le1$, and $\sum_{l:g(l)=i}R_l=\eta_{\mathrm{Nai}}^\dagger U_{\mathbf M}^\dagger\Pi^{(0)}_iU_{\mathbf M}\eta_{\mathrm{Nai}}=M_i$ by Definition~\ref{def:coherent-oracle}. Each call acts on the embedded program half of one port and is measured at once, so the $N$ calls realize $\PP_{\mathbf R}^{\otimes N}$ on the program halves, the setting of Section~\ref{direct} with $D_{\mathrm{Nai}}$ outcomes; Proposition~\ref{prop:direct-rank1} applies, its closed form with $n$ replaced by $D_{\mathrm{Nai}}$ gives~\eqref{eq:measured-naimark}, and $g$ is public, so no record is stored. The sector measurement, the decoder and the relabelling are independent of $\mathbf M$.
\end{proof}

\section{Storage and retrieval via the Kotowski–Oszmaniec representation} \label{kot-osz}

We now consider a second structured oracle model, based on the branch-coherent representation of Kotowski and Oszmaniec~\cite{kotowski2026pretty}, and the exact heralded protocol it supports via probabilistic PBT.

In this resource model we distinguish the information used to \emph{construct} a representation from the information \emph{retained} for the purpose of retrieval. A \emph{designer} who knows $\mathbf M$ supplies a branch ensemble $\{p_\beta,U^{(\beta)},\mathcal Q^{(\beta)}\}_{\beta\in B}$ obeying the Kotowski-Oszmaniec identity~\cite[Result~1]{kotowski2026pretty}, together with a fixed, publicly known projective readout $\{\Pi^{(0)}_l\}$ that does \emph{not} depend on $\mathbf M$. During storage the \emph{device} may sample a branch $\beta\sim\{p_\beta\}$, query its unitary $U^{(\beta)}$ coherently, and sample the associated relabelling $\mathcal Q^{(\beta)}(\cdot\mid l)$ for each readout outcome $l$, but it never receives a classical description of $U^{(\beta)}$. Sampling $\mathcal Q^{(\beta)}(\cdot\mid l)$ for every raw label $l$, including labels not observed in the queries, is granted by the oracle in the same way as sampling a branch $\beta\sim\{p_\beta\}$: the query budget $N$ counts coherent uses of the branch gate, while Proposition~\ref{prop:finite-readout} additionally draws $2d$ conditional classical samples during storage, which are not counted as queries. The receiver is never handed the numerical coefficients $\mathcal Q^{(\beta)}(i\mid l)$, and how the decomposition is obtained from a black-box POVM lies, as before, outside the model.

For fixed $d$, $n$, and $N$ the \emph{retrieval party} is independent of the target POVM: its only measurement-dependent input is a hybrid program consisting of a quantum register and a finite classical record, the embedding and the projective readout being fixed across the family of target POVMs. Crucially, \emph{no} measurement-dependent relabelling library is provided to the receiver for free; Proposition~\ref{prop:finite-readout} shows that a deterministic table with $2d$ entries already suffices as the classical record. In particular the retrieval party never holds a description of $\mathbf M$ and cannot synthesise the measurement directly. We count persistent program memory only: the construction cost of the oracle and any temporary computation during storage lie outside the presented resource model.

\begin{deff}[Branch-coherent measurement oracle $\Omega_{\mathrm{Br}}$]\label{def:branch-oracle}
A branch-coherent oracle $\Omega_{\mathrm{Br}}(\mathbf{M})$ consists of: 
\begin{enumerate}
    \item a finite probability distribution $\{p_{\beta}\}_{\beta\in B}$;
    \item for each $\beta$, a coherent gate $U^{(\beta)}\in\UU(\C^{d}\otimes\C^{2})$;
    \item a fixed publicly known rank-one projective readout $\Pi^{(0)}_{l}=\projj{l}$, $l=1,\dots,2d$, in a fixed orthonormal basis of $\C^{2d}$ (for instance the computational basis of $\C^{d}\otimes\C^{2}$); a coarser projective partition is refined publicly into rank-one outcomes and the relabelling below is lifted to the refinement, which leaves the extended effects~\eqref{eq:extended-compressed} unchanged;
    \item for each $\beta$, a classical stochastic relabelling $\mathcal{Q}^{(\beta)}\!:\{1,\dots,2d\}\to\{1,\dots,n,\emptyset\}$, satisfying the Kotowski-Oszmaniec identity:
        \begin{equation}
            \sum_{\beta}p_{\beta}\,N^{(\beta)}_{i}=q\,M_{i}\quad(i\le n),\qquad \sum_{\beta}p_{\beta}\,N^{(\beta)}_{\emptyset}=(1-q)\,\id_{d},\quad q \in (0,1],
        \end{equation}
 where the \emph{extended effects} on $\C^{2d}$ and their \emph{compression} to $\C^{d}$ are, respectively:
        \begin{equation}\label{eq:extended-compressed}
            E^{(\beta)}_{i}=\sum_{l}\mathcal{Q}^{(\beta)}(i\mid l)\,U^{(\beta)\dagger}\Pi^{(0)}_{l}U^{(\beta)},
            \qquad
            N^{(\beta)}_{i}=\iota^{*}\!\bigl(E^{(\beta)}_{i}\bigr)=(\id_{d}\otimes\bra{0}_{2})\,E^{(\beta)}_{i}\,(\id_{d}\otimes\ket{0}_{2}),
        \end{equation}
    with $\iota^{*}$ the adjoint of the ancilla-preparation channel $\iota(\rho)=\rho\otimes\projj{0}_{2}$.
\end{enumerate}
Every POVM admits such a decomposition with $q=\tfrac18$, independently of $d$~\cite[Result~1]{kotowski2026pretty}; we refer to this as the universal value and use it in the tables. The projective simulation of~\cite{oszmaniec2018simulating} gives a decomposition with $q=\tfrac1d$ and $U^{(\beta)}=V^{(\beta)}\otimes\id_2$ (Remark~\ref{rem:spectral-branch}), the better guarantee for $d<8$. Structured POVMs may admit a larger $q$, up to $q=1$ when a single branch realizes $\mathbf M$ (Appendix~\ref{app:explicit-KO}). The Kotowski-Oszmaniec identity constrains only the compressed effects $N^{(\beta)}_{i}=\iota^{*}(E^{(\beta)}_{i})$, equivalently the action of the extended measurement $\mathbf E^{(\beta)}=\{E^{(\beta)}_{i}\}$ on the embedded subspace $\iota(\BB(\C^{d}))=\{\rho\otimes\projj{0}_{2}\}$; it places no constraint on $E^{(\beta)}_{i}$ elsewhere on $\C^{2d}$. This is inconsequential for exact heralded (pPBT) retrieval: the extended effects are evaluated only on embedded inputs $\rho\otimes\projj{0}_2$, so their action elsewhere is irrelevant.

A query consumes one coherent use of the sampled branch gate $U^{(\beta)}$ on an arbitrary input state of $\C^{2d}$; the fixed readout $\{\Pi^{(0)}_{l}\}$ and the relabelling $\mathcal{Q}^{(\beta)}$ are applied by the experimenter.
The oracle is accessed in one of two modes:
\begin{itemize}
\item[(i)] \emph{ensemble mode:} $\beta$ is sampled once at the start of storage and all $N$ queries use the same $U^{(\beta)}$, so they form the tensor power $\bigl(U^{(\beta)}\bigr)^{\otimes N}$ required by the optimal PSAR compression;
\item[(ii)] \emph{per-query mode:} $\beta$ is sampled independently for each of the $N$ queries.
\end{itemize}
The heralded protocol $\LL_{\mathrm{Br,pPBT}}$ of Definition~\ref{prot Br,pPBT} uses mode~(i).\end{deff}

The Kotowski-Oszmaniec decomposition expresses an arbitrary POVM as a probabilistic mixture of branch measurements acting on $\C^d\otimes\C^2$, with universal non-failure probability $q=\frac18$. Hence, each branch has dimension $2d$, independently of the number of POVM outcomes. Throughout this section we assume access to this decomposition through $\Omega_{\mathrm{Br}}(\mathbf{M})$.

\begin{remark}[Instrument form]\label{rem:branch-instrument}
Averaged over the branch ensemble, the Kotowski-Oszmaniec identity describes a two-outcome instrument $\{\mathcal T_{\textsf s},\mathcal T_{\textsf f}\}$ on $\BB(\C^d)$ carrying an explicit success/failure flag,
\begin{equation}
    \mathcal T_{\textsf s}(\rho)=\sum_{i=1}^{n} q\,\mathrm{Tr}(M_i\rho)\,\projj{i},
    \qquad
    \mathcal T_{\textsf f}(\rho)=(1-q)\,\mathrm{Tr}(\rho)\,\projj{\emptyset},
\end{equation}
with $\mathcal T_{\textsf s}+\mathcal T_{\textsf f}$ trace-preserving; by the Kotowski-Oszmaniec identity, $\mathcal T_{\textsf s}+\mathcal T_{\textsf f}$ is the branch-averaged channel $\rho\mapsto\sum_\beta p_\beta\,\PP_{\mathbf N^{(\beta)}}(\rho)$ of the compressed branch measurements with outcomes in $\{1,\dots,n,\emptyset\}$. The non-failure (success) probability is state-independent, $\mathrm{Tr}[\mathcal T_{\textsf s}(\rho)]=q$ for every $\rho$. The universal bound $q=\tfrac18$ is a worst-case guarantee - for structured POVMs it may be improved to a POVM-dependent value (for the qubit SIC, to $q^{\star}=\tfrac{3-\sqrt3}{2}\approx0.634$ within the two-branch ansatz of Appendix~\ref{app:explicit-KO}).
\end{remark}

\subsection{Exact heralded SAR using pPBT}

Unlike the Naimark route, a single program stores only one sampled branch $\beta$. Hence the target POVM is reproduced exactly only after averaging the sub-normalized outcome probabilities over the sampled branch and readout table and conditioning on the joint teleportation-and-readout success event.

\begin{deff}[The protocol $\LL_{\mathrm{Br,pPBT}}$ for Kotowski-Oszmaniec route]\label{prot Br,pPBT}

Fix the POVM $\mathbf{M}$ and assume access to $\Omega_{\mathrm{Br}}(\mathbf{M})$.

\emph{\textbf{Storage.}}
Sample $\beta\sim\{p_\beta\}$ once at the beginning of the storage phase and fix the corresponding branch unitary $U^{(\beta)}\in\UU(2d)$ for all $N$ oracle queries. Then, conditioned on $\beta$, sample a deterministic relabelling table $f: \{1,\dots,2d\}\rightarrow \{1,\dots,n,\emptyset\}$ according to
\begin{equation}
r_{\beta}(f)=\prod_{l=1}^{2d} \mathcal{Q}^{(\beta)}(f(l)\mid l),
\end{equation}
and store $f$ together with the resulting optimal $N\to1$ PSAR program state for $U^{(\beta)}$. The branch label $\beta$ may then be discarded.

\emph{\textbf{Retrieval.}}
For $\rho\in\Omega(\C^d)$, form the dilated state $\rho\otimes\projj{0}_{2}\in\Omega(\C^{2d})$ and run the pPBT retrieval instrument. On failure flag $\textsf{f}$ abort. On success $\textsf{s}$, the output port carries $U^{(\beta)}(\rho\otimes\projj{0}_{2})U^{(\beta)\dagger}$ exactly. Apply the fixed readout $\{\Pi_l^{(0)}\}_{l=1}^{2d}$, replace the readout label $l$ by $f(l)$, and post-select on a non-failure outcome $i\in\{1,\ldots,n\}$.
\end{deff}

\begin{thm}[Exact heralded SAR of POVMs (averaged over branches and table)]\label{thm:ppbt-exact-kotosz}
Let $\mathbf{M}$ be an arbitrary $n$-outcome POVM on $\C^{d}$, and let $\LL_{\mathrm{Br,pPBT}}$ be the protocol of Definition~\ref{prot Br,pPBT}. Then:
\begin{itemize}
\item[(i)] \emph{Exactness of the averaged program.} Let $\mathcal T^{(\beta,f)}_{\textsf s}$ be the successful subchannel of retrieval from the program storing branch $\beta$ and table $f$, i.e.\ the sub-normalized map to outcomes $i\le n$ on joint pPBT and non-failure success. Averaged over the sampled branch $\beta\sim p_\beta$ and readout table $f\sim r_\beta$,
\begin{equation}\label{eq:branch-table-average}
    \mathbb E_{\beta,f}\bigl[\mathcal T^{(\beta,f)}_{\textsf s}\bigr]=\lambda(N,2d)\,\mathcal T_{\textsf s}=q\,\lambda(N,2d)\,\PP_{\mathbf M},
\end{equation}
with $\mathcal T_{\textsf s}$ the success branch of Remark~\ref{rem:branch-instrument}. Conditioned on joint success, the averaged retrieval therefore reproduces $\mathbf{M}$ exactly, $F_{\mathrm{ch}}\bigl(\hat{\PP}^{\mathrm{cond}}_{\mathbf{M}},\PP_{\mathbf{M}}\bigr)=1$. Exactness is an ensemble statement over independent program preparations; a single program stores one branch and one table and is not guaranteed to reproduce $\mathbf M$ exactly.
\item[(ii)] \emph{Heralding.} Averaged over the sampled branch and table, the joint success probability is state-independent and equals
\begin{equation}
    \Pr[\textsf{s}\wedge \text{non-failure}]\;=\;q\,\lambda(N,2d) \;=\;q\,\frac{N}{N-1+4d^{2}}\;=\;q-\bigt{\frac{q\,d^{2}}{N}}, \quad N \to \infty,
\end{equation}
independent of $n$.
\item[(iii)] \emph{Memory.}  The program register of the optimal-success PSAR protocol has dimension
\begin{equation}
    \dim\HH_M=\binom{N+(2d)^{2}-1}{N}=\bigt{N^{4d^{2}-1}}
\end{equation}
which is polynomial in $N$ at fixed $d$, independent of $n$.
\end{itemize}
\end{thm}

\begin{proof}
For part (i) we observe that, conditioned on pPBT success, the retrieval implements $U^{(\beta)}$ exactly. For a fixed table $f$, the fixed readout followed by $l\mapsto f(l)$ then outputs $i\le n$ with probability $\sum_{l:\,f(l)=i}\mathrm{Tr}\bigl[\Pi^{(0)}_l\,U^{(\beta)}\iota(\rho)U^{(\beta)\dagger}\bigr]$, which in general differs from $\mathrm{Tr}(N^{(\beta)}_i\rho)$. After averaging over the sampled table conditional on the branch, the factorized distribution $r_\beta$ has marginals $\sum_f r_\beta(f)\,\mathbf 1_{\{f(l)=i\}}=\mathcal Q^{(\beta)}(i\mid l)$, so this probability becomes $\mathrm{Tr}(N^{(\beta)}_i\rho)$ by~\eqref{eq:extended-compressed}. Averaging over $\beta$ and using Definition~\ref{def:branch-oracle} gives $\sum_\beta p_\beta\mathrm{Tr}(N^{(\beta)}_i\rho)=q\,\mathrm{Tr}(M_i\rho)$; together with the pPBT success probability $\lambda(N,2d)$ this is~\eqref{eq:branch-table-average}. Summing over $i\le n$ gives the averaged joint success probability $q\,\lambda(N,2d)$, so conditioning the averaged retrieval on joint success gives $\Pr[i\mid\text{joint success}]=\mathrm{Tr}(M_i\rho)$, which is precisely $F_{\mathrm{ch}}=1$ in expectation.
The conditioning is applied to the sub-normalized probabilities averaged over branch and table rather than to the conditional distributions of individual programs - as this ratio is nonlinear, exactness is not guaranteed for an individual program encoding a fixed branch and table.

For part (ii) one can notice that the pPBT success probability $\lambda(N,2d)$ is independent of $\beta$, $f$ and $\rho$. As a consequence, averaged over $\beta$ and $f$, the joint success probability is $q\,\lambda(N,2d)$. For last point, Theorem~\ref{PSAR} with $D=2d$ gives $\dim\HH_M=\binom{N+(2d)^{2}-1}{N}$ which concludes the proof.
\end{proof}

The branch-coherent representation therefore provides an alternative route to programmable measurement retrieval. Compared with the Naimark construction, it replaces the dependence on the dilation dimension by the fixed branch dimension $2d$, and compared with the compressed direct program, it replaces a polynomial degree growing with $n$ by $4d^2-1$. In both cases, these gains are counterbalanced by the intrinsic success probability $q$ of the Kotowski-Oszmaniec decomposition.

\subsection{Finite classical readout without a shared library}

\begin{propo}[A finite classical record for branch retrieval]
\label{prop:finite-readout}
Fix $d\ge2$, $n\ge1$, and $N\ge1$, and suppose the branch oracle of Definition~\ref{def:branch-oracle} provides a decomposition with success probability $q$ and a fixed readout with $2d$ outcomes, and that during storage the associated relabelling can be sampled for every readout outcome. Then there exists a heralded retrieval protocol whose persistent program is a quantum register of dimension
\begin{equation}
    L=\binom{N+4d^2-1}{N}
\end{equation}
together with a deterministic table $f:\{1,\dots,2d\}\to\{1,\dots,n,\emptyset\}$; equivalently
\begin{equation}
 Q=\left\lceil\log_2 L\right\rceil,\qquad C=2d\left\lceil\log_2(n+1)\right\rceil .
\end{equation}
Averaged over the sampled branch and table, and conditioned on the joint success event, the protocol retrieves $\mathbf M$ exactly, with state-independent success probability
\begin{equation}
 p_{\rm succ}=q\,\lambda(N,2d)=q\,\frac{N}{N-1+4d^2}.
\end{equation}
The retrieval is \emph{self-contained} - the receiver holds only the quantum register and the table $f$, needing neither the branch label $\beta$ nor the relabelling library $\{\mathcal Q^{(\beta)}\}_{\beta\in B}$. Bundled into one register, the program has hybrid carrier dimension at most $L(n+1)^{2d}$.
\end{propo}

\begin{proof}
Sample $\beta\sim\{p_\beta\}$ and, conditioned on $\beta$, sample a table $f$ with probability $r_\beta(f)=\prod_{l=1}^{2d}\mathcal Q^{(\beta)}(f(l)\mid l)$. This is a probability distribution on tables, and since it factorizes over $l$ its marginals reproduce the relabelling exactly,
\begin{equation}
 \sum_f r_\beta(f)\,\mathbf 1_{\{f(l)=i\}}=\mathcal Q^{(\beta)}(i\mid l).
\end{equation}
Store the $N$-use PSAR program for $U^{(\beta)}$ (in dimension $2d$) together with $f$, and discard $\beta$. At retrieval, run the universal unitary-PSAR processor, perform the fixed readout $\{\Pi^{(0)}_l\}$, and replace its outcome $l$ by $f(l)$; declare failure if PSAR fails or $f(l)=\emptyset$. Conditioned on PSAR success (probability $\lambda(N,2d)$) the retrieved unitary is exact, so the sub-normalized probability of a non-failure output $i\le n$, averaged over $f$ and then over $\beta$, is
\begin{equation}
    \begin{split}
        & \lambda(N,2d)\sum_\beta p_\beta\sum_{l}\mathcal Q^{(\beta)}(i\mid l)\, \mathrm{Tr}\!\bigl[\Pi^{(0)}_l\,U^{(\beta)}\iota(\rho)U^{(\beta)\dagger}\bigr]=\lambda(N,2d)\sum_\beta p_\beta\,\mathrm{Tr}\!\bigl(N^{(\beta)}_i\rho\bigr)\\
        &=q\,\lambda(N,2d)\,\mathrm{Tr}(M_i\rho),
    \end{split}
\end{equation}
using the effects of~\eqref{eq:extended-compressed} and the Kotowski-Oszmaniec identity. Hence, averaged over the sampled branch and table, the successful map is $q\,\lambda(N,2d)\,\PP_{\mathbf M}$, as in~\eqref{eq:branch-table-average}.

Since $f$ prescribes the relabelling $l\mapsto f(l)$ directly and the PSAR processor and readout do not depend on $\beta$, the label $\beta$ may be discarded. Finally, there are $(n+1)^{2d}$ tables, so storing each of the $2d$ entries in a $\lceil\log_2(n+1)\rceil$-bit register gives $C=2d\lceil\log_2(n+1)\rceil$ (a joint encoding would use $\lceil2d\log_2(n+1)\rceil$ bits); the logical hybrid dimension is at most $L(n+1)^{2d}$.
\end{proof}

\begin{propo}[Measured-branch program]\label{prop:measured-branch}
Let the readout of $\Omega_{\mathrm{Br}}$ be $\Pi^{(0)}_l=\projj{l}$, $l=1,\dots,2d$, and let $\eta:\C^d\hookrightarrow\C^{2d}$ be the fixed embedding. For a sampled branch $\beta$ the operators $R^{(\beta)}_l:=\eta^\dagger U^{(\beta)\dagger}\projj{l}U^{(\beta)}\eta$ form a $2d$-outcome POVM on $\C^d$ whose effects have rank at most one. Sample $\beta$ and the table $f$ as in Proposition~\ref{prop:finite-readout}; apply each of the $N$ coherent calls to an embedded $d$-dimensional port of the optimal pPBT resource, followed immediately by the readout; encode the resulting $N$ uses of $\mathbf R^{(\beta)}$ with Proposition~\ref{prop:direct-rank1}, retaining the raw-outcome type and $f$. At retrieval, run the direct protocol and relabel its output by $f$, rejecting $\emptyset$. Then, with exactly $N$ coherent queries,
\begin{equation}\label{eq:measured-branch}
\mathbb E_{\beta,f}\,\mathcal T^{(\beta,f)}_{\mathsf s}=q\,\lambda(N,d)\,\PP_{\mathbf M},
\end{equation}
and the persistent memory is
\begin{equation}\label{eq:measured-branch-memory}
Q=\Bigl\lceil\log_2\max_{t\in T_{2d,N}}\prod_{l=1}^{2d}\binom{t_l+d-1}{d-1}\Bigr\rceil,\qquad
C=\Bigl\lceil\log_2\binom{N+2d-1}{2d-1}\Bigr\rceil+2d\lceil\log_2(n+1)\rceil .
\end{equation}
The logical direct-sum dimension is $(n+1)^{2d}\binom{N+2d^2-1}{N}$, of degree $2d^2-1$ in $N$, against $4d^2-1$ for Proposition~\ref{prop:finite-readout}.
\end{propo}
\begin{proof}
\emph{The readout is rank one.} By Definition~\ref{def:branch-oracle} the readout is $\Pi^{(0)}_l=\projj l$ for a fixed orthonormal basis $\{\ket l\}$ of $\C^{2d}$; a coarser partition is first refined publicly into rank-one outcomes with the relabelling lifted, which changes neither the extended effects nor~\eqref{eq:extended-compressed}.

\emph{The raw POVM.} Write $\eta:=\id_d\otimes\ket0_2:\C^d\to\C^{2d}$, so that $\iota^*(X)=\eta^\dagger X\eta$. Then $\sum_lR^{(\beta)}_l=\eta^\dagger U^{(\beta)\dagger}U^{(\beta)}\eta=\eta^\dagger\eta=\id_d$ and $\operatorname{rank}R^{(\beta)}_l\le\operatorname{rank}\Pi^{(0)}_l=1$, so $\mathbf R^{(\beta)}$ is a $2d$-outcome POVM on $\C^d$ with effects of rank at most one; some $R^{(\beta)}_l$ may vanish. By~\eqref{eq:extended-compressed}, $\sum_l\mathcal Q^{(\beta)}(i\mid l)R^{(\beta)}_l=\iota^*\bigl(E^{(\beta)}_i\bigr)=N^{(\beta)}_i$.

\emph{Admissibility.} In ensemble mode, $\beta$ is sampled once and all $N$ calls use $U^{(\beta)}$. The $k$-th call applies $U^{(\beta)}$ to $\eta(B_k)$, the embedded program half of the $k$-th port of the optimal pPBT resource, which is a state of $\C^{2d}$ as Definition~\ref{def:branch-oracle} allows, and is followed by the readout $\{\Pi^{(0)}_l\}$. Each call therefore realizes one use of $\PP_{\mathbf R^{(\beta)}}$ on $B_k$, and exactly $N$ coherent queries are consumed.

\emph{Compression and retrieval.} The $N$ uses of the fixed POVM $\mathbf R^{(\beta)}$ on $\C^d$ place us in the setting of Section~\ref{direct} with $2d$ outcomes; Proposition~\ref{prop:direct-rank1} applies, since its effects have rank at most one, and yields the program of~\eqref{eq:measured-branch-memory} (the table $f$ added to the classical record), together with exact reproduction of $\PP_{\mathbf R^{(\beta)}}$ on success with probability $\lambda(N,d)$. The output $l$ is relabelled to $f(l)$, and $\emptyset$ is rejected.

\emph{Averaging.} For fixed $\beta,f$ the successful subchannel for $i\le n$ is $\rho\mapsto\lambda(N,d)\sum_{l:f(l)=i}\Tr(R^{(\beta)}_l\rho)$. Averaging over $f\sim r_\beta$, whose $l$-th marginal is $\mathcal Q^{(\beta)}(\cdot\mid l)$, gives $\lambda(N,d)\sum_l\mathcal Q^{(\beta)}(i\mid l)\Tr(R^{(\beta)}_l\rho)=\lambda(N,d)\Tr(N^{(\beta)}_i\rho)$; averaging over $\beta$ and using the Kotowski--Oszmaniec identity gives~\eqref{eq:measured-branch}.

\emph{Dimension.} By the closed form in Proposition~\ref{prop:direct-rank1} with $n$ replaced by $2d$, $\sum_{t\in T_{2d,N}}\prod_{l}\binom{t_l+d-1}{d-1}=\binom{N+2d\cdot d-1}{N}$; there are $(n+1)^{2d}$ tables.
\end{proof}

\begin{cor}[Measured-branch program with quantum type register]\label{cor:measured-branch-quantum}
Storing the raw-outcome type in the quantum register as in~\eqref{eq:direct-rank1-quantum} turns the program of Proposition~\ref{prop:measured-branch} into one with
\begin{equation}\label{eq:measured-branch-quantum}
Q=\Bigl\lceil\log_2\binom{N+2d^2-1}{N}\Bigr\rceil,\qquad C=2d\lceil\log_2(n+1)\rceil,\qquad \mathbb E_{\beta,f}\,\mathcal T^{(\beta,f)}_{\mathsf s}=q\,\lambda(N,d)\,\PP_{\mathbf M}.
\end{equation}
Compared with Proposition~\ref{prop:finite-readout}, which stores $\binom{N+4d^2-1}{N}$ dimensions with the same classical record and success $q\lambda(N,2d)$, this program has a strictly smaller quantum register and strictly higher success for every $N\ge1$, and its quantum register is likewise independent of $n$.
\end{cor}
\begin{proof}
Apply the quantum-type-register allocation of Proposition~\ref{prop:direct-rank1}, with $n$ replaced by $2d$, to the program of Proposition~\ref{prop:measured-branch}; the table $f$ stays classical. The sector measurement at retrieval, the decoder and the relabelling are independent of $\mathbf M$, so~\eqref{eq:measured-branch} is unchanged, and the dimension is the closed form $\sum_{t\in T_{2d,N}}\prod_l\binom{t_l+d-1}{d-1}=\binom{N+2d^2-1}{N}$. For the comparison, $\binom{N+2d^2-1}{N}<\binom{N+4d^2-1}{N}$ for $N\ge1$ and $\lambda(N,d)>\lambda(N,2d)$.
\end{proof}

\begin{remark}
The two allocations of the measured-branch program trade quantum against classical memory: with a classical type record the quantum register is the smaller of the two but the record grows as $(2d-1)\log_2N+O_d(1)$; with the type in the quantum register the record is the $N$-independent table of $2d\lceil\log_2(n+1)\rceil$ bits. In either allocation the measured-branch program has higher success $q\lambda(N,d)>q\lambda(N,2d)$ and lower polynomial degree than the full-unitary program of Proposition~\ref{prop:finite-readout}, which with the same classical record uses the larger register $\binom{N+4d^2-1}{N}$ (Table~\ref{tab:branch}). The full-unitary program is therefore dominated as an implementation of the measurement task; we retain it as the reference construction obtained by importing unitary PSAR, which solves the stronger task of retrieving the branch unitary itself.
\end{remark}

\begin{remark}
As in Theorem~\ref{thm:ppbt-exact-kotosz}, exactness is a property of the full randomized protocol - after averaging the sub-normalized maps over the sampled branch and table, before conditioning - not of any single stored table $f$. The cost $C=2d\lceil\log_2(n+1)\rceil$ is independent of the number of branches $|B|$, which may be arbitrarily large: only one sampled table reaches the receiver.
\end{remark}

\begin{remark}[No-ancilla branch program]\label{rem:spectral-branch}
When the designer may choose the decomposition, the projective simulation of~\cite{oszmaniec2018simulating} supplies a branch oracle with $q=\tfrac1d$: write $M_i=\sum_a\alpha_{ia}\projj{\psi_{ia}}$ with $\sum_{i,a}\alpha_{ia}=d$, sample $\beta=(i,a)$ with probability $\alpha_{ia}/d$, let $U^{(\beta)}=V_{ia}\otimes\id_2$ with $V_{ia}\in\UU(d)$ mapping $\psi_{ia}$ to the first basis vector, and let the relabelling accept the raw outcome $(1,0)$ as $i$ and reject all others; then $\sum_\beta p_\beta N^{(\beta)}_i=M_i/d$. With this structure promised publicly, the measured-branch program of Corollary~\ref{cor:measured-branch-quantum} encodes the $d$-outcome raw readout $\{V_{ia}^\dagger\projj mV_{ia}\}_{m\le d}$ and stores only the label $i$:
\begin{equation}\label{eq:spectral-branch}
Q=\Bigl\lceil\log_2\binom{N+d^2-1}{N}\Bigr\rceil,\qquad C=\lceil\log_2n\rceil,\qquad \mathbb E_{\beta}\,\mathcal T^{(\beta)}_{\mathsf s}=\tfrac1d\,\lambda(N,d)\,\PP_{\mathbf M}.
\end{equation}
For the qubit SIC POVM at $N=50$ this gives $(Q,C,P)=(15,2,0.472)$, against $(28,12,0.118)$ for the universal choice $q=\tfrac18$ in Table~\ref{tab:branch}. It is a different supplied decomposition, not a conversion of an arbitrary branch oracle, the two guarantees coincide at $d=8$, and the universal value is the larger one for $d>8$.
\end{remark}

\section{Memory and success at a glance}\label{comparison}

The exact constructions of Sections~\ref{direct}, \ref{naimark} and~\ref{kot-osz} all achieve $F_{\mathrm{ch}}\bigl(\hat{\PP}^{\mathrm{cond}}_{\mathbf M},\PP_{\mathbf M}\bigr)=1$ (for the branch route with the program averaged over branch and table as in Definition~\ref{def:exact-retrieval}, Theorem~\ref{thm:ppbt-exact-kotosz}); they differ in success probability and in quantum and classical memory. Table~\ref{tab:numerical} holds the number of queries $N$ fixed and reports, for each construction, the number $Q$ of qubits and $C$ of bits of persistent program memory together with the success probability, from the closed forms of Propositions~\ref{prop:direct-ppbt}, \ref{prop:direct-cq}, \ref{prop:direct-rank1} and~\ref{prop:finite-readout} and Theorem~\ref{thm:ppbt-exact-naimark}, for binary qubit projective measurements, the qubit SIC POVM~\cite{renes2004symmetric}, qubit POVMs with six and eight outcomes, and the qutrit SIC POVM. All tabulated POVMs are rank-one, so the Naimark route is listed with the minimal dilation $D_{\min}=n$ (Remark~\ref{rem:naimark-dim}); the canonical dilation $D_{\mathrm{Nai}}=dn$ appears in Table~\ref{tab:memoryQC}. The quantum register of each compressed direct entry was found by exact maximization over all types $t\in T_{n,N}$; for $d=2$ the balanced type is optimal because $K_2(m)=(\lfloor m/2\rfloor+1)(\lceil m/2\rceil+1)$ is log-concave, and the binomial sequence $m\mapsto\binom{m+d-1}{d-1}$ of the rank-one column is log-concave for every $d$. We do not claim log-concavity of $K_d$ in general. Table~\ref{tab:memoryQC} splits the same programs into qubits and bits. Throughout, $Q=\lceil\log_2\dim\rceil$ counts the qubits of the quantum register and $C$ the bits of the classical record, so every column measures the same resource. $\mathrm{Nai}_{\min}$ assumes that all $n$ effects are nonzero and rank one, $D_{\min}=n$, and stores only fixed side information; the branch columns use the universal $q=\tfrac18$, which structured POVMs improve on (Section~\ref{sec:discussion}). Table~\ref{tab:benchmark} isolates the central comparison for one measurement.

\begin{table}[H]
\centering
{\small\setlength{\tabcolsep}{8pt}
\begin{tabular}{l c c c}
\toprule
Qubit SIC POVM ($d=2$, $n=4$), $N=50$ queries & $Q$ [qubits] & $C$ [bits] & success $P$\\
\midrule
Direct, rank-one encoder, type classical (Proposition~\ref{prop:direct-rank1}) & $16$ & $15$ & $0.9434$\\
Direct, rank-one encoder, type quantum~\eqref{eq:direct-rank1-quantum} & $28$ & $0$ & $0.9434$\\
Minimal Naimark, full-unitary program (Theorem~\ref{thm:ppbt-exact-naimark}) & $48$ & $0$ & $0.7692$\\
\bottomrule
\end{tabular}}
\caption{The central comparison, evaluated for one measurement: the same rank-one direct program in its two allocations of the type record, against the optimal-success PSAR program of the minimal Naimark unitary. At this query count the measurement program uses less memory and succeeds more often than the full-unitary program, and the type record can be traded between the two registers.}
\label{tab:benchmark}
\end{table}

\begin{table}[H]
\centering
{\footnotesize\setlength{\tabcolsep}{2.6pt}
\begin{tabular}{ccc cc cc cc c cc ccc}
\toprule
& & & \multicolumn{2}{c}{Direct} & \multicolumn{2}{c}{Direct, compr.} & \multicolumn{2}{c}{Direct, rank-one} & $\mathrm{Nai}_{\min}$ & \multicolumn{2}{c}{Branch} & \multicolumn{3}{c}{Heralded success $P$}\\
& & & \multicolumn{2}{c}{uncompressed} & \multicolumn{2}{c}{Prop.~\ref{prop:direct-cq}} & \multicolumn{2}{c}{Prop.~\ref{prop:direct-rank1}} & & \multicolumn{2}{c}{$q=\tfrac18$} & & & \\
\cmidrule(lr){4-5}\cmidrule(lr){6-7}\cmidrule(lr){8-9}\cmidrule(lr){10-10}\cmidrule(lr){11-12}\cmidrule(lr){13-15}
$d$ & $n$ & $N$ & $Q$ & $C$ & $Q$ & $C$ & $Q$ & $Q^{\mathrm q}$ & $Q$ & $Q$ & $C$ & Direct & $\mathrm{Nai}_{\min}$ & Branch\\
\midrule
$2$ & $2$ & $10$ & $10$ & $10$ & $8$ & $4$ & $6$ & $9$ & $9$ & $22$ & $8$ & $0.769$ & $0.769$ & $0.050$ \\
$2$ & $2$ & $50$ & $50$ & $50$ & $16$ & $6$ & $10$ & $15$ & $15$ & $48$ & $8$ & $0.943$ & $0.943$ & $0.096$ \\
$2$ & $2$ & $100$ & $100$ & $100$ & $19$ & $7$ & $12$ & $18$ & $18$ & $62$ & $8$ & $0.971$ & $0.971$ & $0.109$ \\
\addlinespace
$2$ & $4$ & $10$ & $10$ & $20$ & $10$ & $9$ & $8$ & $15$ & $22$ & $22$ & $12$ & $0.769$ & $0.400$ & $0.050$ \\
$2$ & $4$ & $50$ & $50$ & $100$ & $23$ & $15$ & $16$ & $28$ & $48$ & $48$ & $12$ & $0.943$ & $0.769$ & $0.096$ \\
$2$ & $4$ & $100$ & $100$ & $200$ & $31$ & $18$ & $19$ & $35$ & $62$ & $62$ & $12$ & $0.971$ & $0.870$ & $0.109$ \\
\addlinespace
$2$ & $6$ & $10$ & $10$ & $26$ & $10$ & $12$ & $9$ & $19$ & $32$ & $22$ & $12$ & $0.769$ & $0.222$ & $0.050$ \\
$2$ & $6$ & $50$ & $50$ & $130$ & $29$ & $22$ & $20$ & $39$ & $80$ & $48$ & $12$ & $0.943$ & $0.588$ & $0.096$ \\
$2$ & $6$ & $100$ & $100$ & $259$ & $39$ & $27$ & $25$ & $49$ & $108$ & $62$ & $12$ & $0.971$ & $0.741$ & $0.109$ \\
\addlinespace
$2$ & $8$ & $10$ & $10$ & $30$ & $10$ & $15$ & $10$ & $22$ & $40$ & $22$ & $16$ & $0.769$ & $0.137$ & $0.050$ \\
$2$ & $8$ & $50$ & $50$ & $150$ & $33$ & $28$ & $23$ & $48$ & $109$ & $48$ & $16$ & $0.943$ & $0.442$ & $0.096$ \\
$2$ & $8$ & $100$ & $100$ & $300$ & $46$ & $35$ & $31$ & $62$ & $153$ & $62$ & $16$ & $0.971$ & $0.613$ & $0.109$ \\
\addlinespace
$3$ & $9$ & $10$ & $16$ & $32$ & $16$ & $16$ & $16$ & $28$ & $43$ & $32$ & $24$ & $0.556$ & $0.111$ & $0.028$ \\
$3$ & $9$ & $50$ & $80$ & $159$ & $59$ & $31$ & $42$ & $68$ & $122$ & $80$ & $24$ & $0.862$ & $0.385$ & $0.074$ \\
$3$ & $9$ & $100$ & $159$ & $317$ & $89$ & $39$ & $57$ & $90$ & $175$ & $108$ & $24$ & $0.926$ & $0.556$ & $0.093$ \\
\bottomrule
\end{tabular}}
\caption{Persistent program memory $(Q,C)$ and heralded success $P$ at fixed query count $N$, notation as defined above. Direct, uncompressed: Proposition~\ref{prop:direct-ppbt}. Direct, compressed: Proposition~\ref{prop:direct-cq}. Direct, rank-one: Proposition~\ref{prop:direct-rank1}, $Q$ with the type kept classical (same $C$ as the compressed column) and $Q^{\mathrm q}$ with the type in the quantum register ($C=0$). $\mathrm{Nai}_{\min}$: minimal Naimark dilation, $Q=\lceil\log_2\binom{N+n^2-1}{N}\rceil$. Branch: Proposition~\ref{prop:finite-readout}. The three direct columns share the success $\lambda(N,d)$.}
\label{tab:numerical}
\end{table}

At fixed $N$ the direct route has the largest success probability, $\lambda(N,d)$, because it teleports in dimension $d$: $\lambda(N,d)\ge\lambda(N,D_{\mathrm{Nai}})$ with equality only when $D_{\mathrm{Nai}}=d$, and $q\,\lambda(N,2d)<\lambda(N,d)$; the measured-branch program reaches $q\lambda(N,d)$ and ties with the direct route when $q=1$ (Section~\ref{sec:discussion}). For POVMs with $n$ nonzero rank-one effects, $D_{\min}=n$, the minimal Naimark route has success $\frac{N}{N-1+n^2}$ and register $\binom{N+n^2-1}{N}$, the canonical one $\frac{N}{N-1+d^2n^2}$ and $\binom{N+d^2n^2-1}{N}$; the branch register $\binom{N+4d^2-1}{N}$ does not depend on $n$, at the price of the factor $q$. For the qubit SIC the minimal Naimark and branch registers coincide because $D_{\min}=n=4=2d$, and the branch success at $q=\tfrac18$ is only the worst-case guarantee: a single padded branch realizes the SIC with $q=1$ on the same register, and the two-branch ansatz of Appendix~\ref{app:explicit-KO} gives $q^\star\approx0.634$. More generally, for rank-one POVMs with $n\le2d$ one padded branch realizes the POVM with $q=1$. For the qutrit SIC at $N=50$ even the uncompressed direct register is smaller than the minimal Naimark one: at this budget the polynomial register is the larger one.

The direct program admits two allocations of its type record: with the type kept classical it has the smallest quantum register in Table~\ref{tab:numerical}; with the type in the quantum register~\eqref{eq:direct-rank1-quantum} it is purely quantum, like the Naimark programs. For the benchmark families $D_{\min}=n\ge d$ and $nd\le n^2$, so the purely quantum direct program is never larger than the minimal Naimark register and never less successful, strictly so whenever $n>d$ (columns $Q^{\mathrm q}$ and $\mathrm{Nai}_{\min}$). This program is also available under the Naimark oracle itself, and Proposition~\ref{prop:measured-naimark} extends the comparison to targets of any rank: the full-unitary Naimark registers are dominated as implementations of the measurement, although the task they solve, retrieving the dilation unitary on its whole space, is stronger. A purely quantum program is thus not a benefit specific to storing the Naimark unitary; what the comparison turns on is where the type information is kept.

\begin{table}[H]
\centering
\resizebox{\linewidth}{!}{\setlength{\tabcolsep}{6pt}
\begin{tabular}{l c c}
\toprule
Scheme & $Q$ [qubits] & $C$ [bits] \\
\midrule
Direct uncompressed pPBT
 & $\left\lceil N\log_2 d\right\rceil$
 & $\leq\left\lceil N\log_2 n\right\rceil$\\[4pt]

Direct compressed pPBT
& $\left\lceil \log_2 \max_{t\in T_{n,N}} \prod_{j=1}^{n} K_d(t_j) \right\rceil$
& $\left\lceil \log_2\binom{N+n-1}{n-1} \right\rceil$\\[4pt]
Direct compressed pPBT, rank-one effects
& $\left\lceil \log_2 \max_{t\in T_{n,N}} \prod_{j=1}^{n} \binom{t_j+d-1}{d-1} \right\rceil$
& $\left\lceil \log_2\binom{N+n-1}{n-1} \right\rceil$\\[4pt]
Direct compressed pPBT, rank-one effects, type in quantum register
& $\left\lceil \log_2\binom{N+nd-1}{N} \right\rceil$
& $0$\\[4pt]

Naimark $+$ pPBT
 & $\left\lceil \log_2\binom{N+(dn)^{2}-1}{N}\right\rceil$
 & fixed side information \\[4pt]

Naimark minimal (rank-1) $+$ pPBT
 & $\left\lceil \log_2\binom{N+n^{2}-1}{N} \right\rceil$
 & fixed side information \\[4pt]
Measured Naimark $+$ pPBT, type in quantum register
 & $\left\lceil \log_2\binom{N+dD_{\mathrm{Nai}}-1}{N} \right\rceil$
 & fixed side information \\[4pt]

Branch $+$ pPBT
 & $\left\lceil \log_2\binom{N+4d^{2}-1}{N} \right\rceil$
 & $2d\left\lceil\log_2(n+1)\right\rceil$ \\[4pt]
Measured branch $+$ pPBT, type classical
 & $\left\lceil \log_2 \max_{t\in T_{2d,N}} \prod_{l=1}^{2d} \binom{t_l+d-1}{d-1} \right\rceil$
 & $\left\lceil \log_2\binom{N+2d-1}{2d-1} \right\rceil+2d\left\lceil\log_2(n+1)\right\rceil$ \\[4pt]
Measured branch $+$ pPBT, type in quantum register
 & $\left\lceil \log_2\binom{N+2d^{2}-1}{N} \right\rceil$
 & $2d\left\lceil\log_2(n+1)\right\rceil$ \\
\bottomrule
\end{tabular}}
\caption{The SAR schemes for POVMs and their memory. Quantum memory $Q$ is measured in qubits and classical memory $C$ in bits. The classical-memory costs assume the outcome string is encoded jointly, giving $\lceil N\log_2 n\rceil$ bits; storing one separate register per query would instead cost $N\lceil\log_2 n\rceil$ bits. No claim of optimal compression is made. For the Naimark routes, the side information consists of the dilation dimension $D_{\mathrm{Nai}}$, the fixed embedding isometry $\eta_{\mathrm{Nai}}$, and the fixed projective readout $\{\Pi^{(0)}_i\}$, all fixed across the target family (for the minimal realisation this fixes the rank profile of the effects, Definition~\ref{def:coherent-oracle}) and none of which depends on the unknown $U_{\mathbf M}$ (no classical description of $U_{\mathbf M}$ is available, otherwise the measurement could be synthesised directly). }
\label{tab:memoryQC}
\end{table}

\begin{table}[H]
\centering
{\small\setlength{\tabcolsep}{5pt}
\begin{tabular}{ccc ccc ccc cc}
\toprule
& & & \multicolumn{3}{c}{Full-unitary branch} & \multicolumn{3}{c}{Measured branch, type classical} & \multicolumn{2}{c}{Measured branch, type quantum}\\
& & & \multicolumn{3}{c}{Prop.~\ref{prop:finite-readout}} & \multicolumn{3}{c}{Prop.~\ref{prop:measured-branch}} & \multicolumn{2}{c}{Cor.~\ref{cor:measured-branch-quantum}}\\
\cmidrule(lr){4-6}\cmidrule(lr){7-9}\cmidrule(lr){10-11}
$d$ & $n$ & $N$ & $Q$ & $C$ & $P$ & $Q$ & $C$ & $P$ & $Q$ & $C$\\
\midrule
$2$ & $2$ & $50$ & $48$ & $8$ & $0.0962$ & $16$ & $23$ & $0.1179$ & $28$ & $8$ \\
$2$ & $2$ & $200$ & $76$ & $8$ & $0.1163$ & $23$ & $29$ & $0.1232$ & $42$ & $8$ \\
\addlinespace
$2$ & $4$ & $50$ & $48$ & $12$ & $0.0962$ & $16$ & $27$ & $0.1179$ & $28$ & $12$ \\
$2$ & $4$ & $200$ & $76$ & $12$ & $0.1163$ & $23$ & $33$ & $0.1232$ & $42$ & $12$ \\
\addlinespace
$2$ & $6$ & $50$ & $48$ & $12$ & $0.0962$ & $16$ & $27$ & $0.1179$ & $28$ & $12$ \\
$2$ & $6$ & $200$ & $76$ & $12$ & $0.1163$ & $23$ & $33$ & $0.1232$ & $42$ & $12$ \\
\addlinespace
$2$ & $8$ & $50$ & $48$ & $16$ & $0.0962$ & $16$ & $31$ & $0.1179$ & $28$ & $16$ \\
$2$ & $8$ & $200$ & $76$ & $16$ & $0.1163$ & $23$ & $37$ & $0.1232$ & $42$ & $16$ \\
\addlinespace
$3$ & $9$ & $50$ & $80$ & $24$ & $0.0735$ & $34$ & $46$ & $0.1078$ & $52$ & $24$ \\
$3$ & $9$ & $200$ & $139$ & $24$ & $0.1064$ & $56$ & $56$ & $0.1202$ & $83$ & $24$ \\
\bottomrule
\end{tabular}}
\caption{Three exact programs under the same branch oracle, universal choice $q=\tfrac18$. All average the successful subchannel over branch and table. The measured-branch program has the higher success $q\lambda(N,d)$ in both of its allocations (the last group has the same $P$ as the middle one): with a classical type record its quantum register is the smallest but the record grows with $N$; with the type in the quantum register its record is the same $N$-independent table of $2d\lceil\log_2(n+1)\rceil$ bits as the full-unitary program's, with a smaller register. The full-unitary program is dominated in every row.}
\label{tab:branch}
\end{table}

Figure~\ref{fig:qubits_vs_N} holds the POVM fixed (qubit SIC) and varies $N$. Stored uncompressed, the direct program uses linearly many qubits; the PSAR-based registers and the compressed direct program use logarithmically many, and the compressed direct program needs fewer qubits than the minimal-Naimark and full-unitary branch registers at every $N\le500$, even with its classical record counted, as does its purely quantum allocation~\eqref{eq:direct-rank1-quantum}. Holding the quantum-memory budget fixed instead of $N$ does not change the ordering: at any common budget in this range the compressed direct program fits more queries and succeeds with higher probability. The measured-branch program is not separated from the direct one here: its raw readout has $2d=n$ outcomes, so its registers coincide with the rank-one direct registers, and with the single-branch realization $q=1$ its success coincides as well. For $d=2$ and every target success below $\tfrac18$ a single direct query already suffices, with one qubit and $\lceil\log_2 n\rceil$ bits (Remark~\ref{rem:direct-outcome-scaling}).

\begin{figure}[H]
    \centering
    \includegraphics[width=\linewidth]{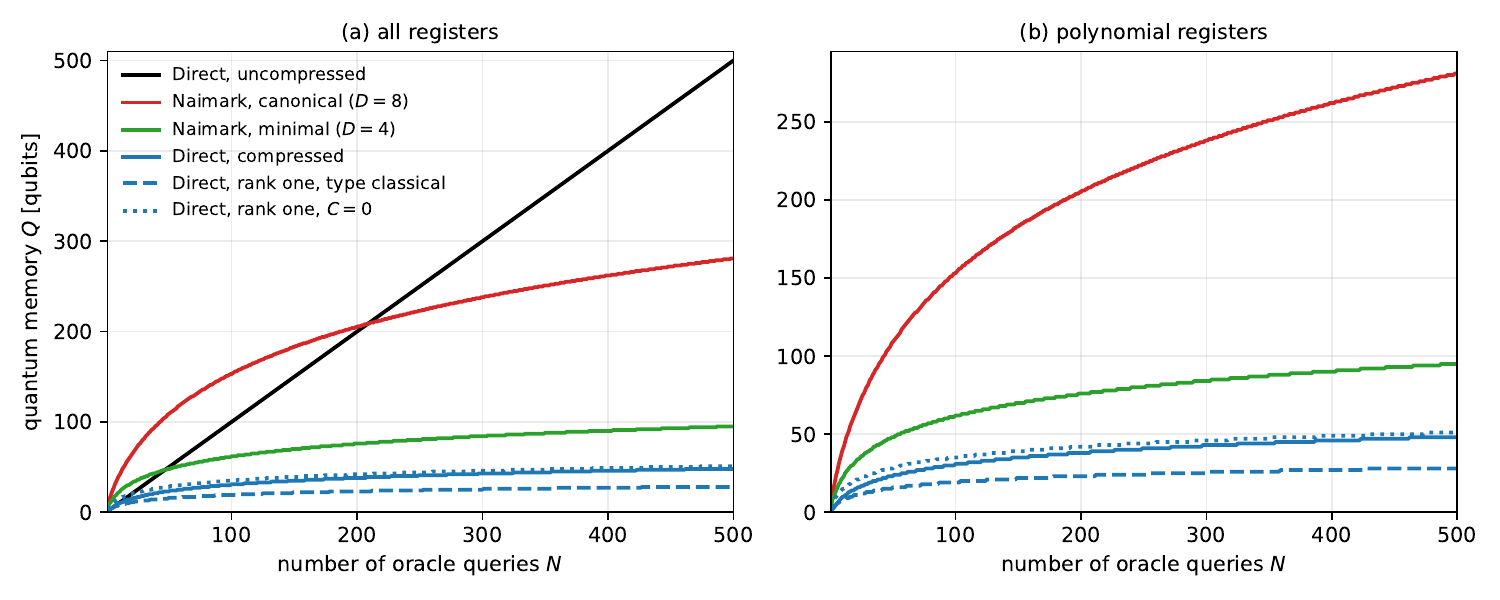}
   \caption{Quantum memory $Q$ in qubits as a function of the number of oracle queries $N$ for the qubit SIC POVM ($d=2$, $n=4$), exact constructions only: the uncompressed direct program ($\lceil N\log_2 d\rceil$, Proposition~\ref{prop:direct-ppbt}), the compressed direct program (Proposition~\ref{prop:direct-cq}), its rank-one refinement with the type kept classical and with the type in the quantum register (Proposition~\ref{prop:direct-rank1}), the canonical and minimal Naimark registers (Theorem~\ref{thm:ppbt-exact-naimark}), and the branch pPBT register (Theorem~\ref{thm:ppbt-exact-kotosz}). The minimal Naimark and branch pPBT curves coincide because $D_{\min}=n=4=2d$; the measured-branch registers of Proposition~\ref{prop:measured-branch} and Corollary~\ref{cor:measured-branch-quantum} coincide with the two rank-one direct curves because the raw readout has $2d=n$ outcomes. Panel~(b) resolves the polynomial registers. For classical memory see Table~\ref{tab:numerical}.}
    \label{fig:qubits_vs_N}
\end{figure}

Which route is available is decided by the oracle: direct access needs only the POVM channel; the Naimark route presupposes coherent access to a dilation unitary on the full dilation space; the branch route a decomposition supplied by a designer (Section~\ref{sec:disc-programs}). Neither a purely quantum program nor an $N$-independent record is exclusive to the structured routes, since the direct program can be stored purely quantumly~\eqref{eq:direct-rank1-quantum} and the measured-branch program admits an $N$-independent record (Corollary~\ref{cor:measured-branch-quantum}). In the examples with $n>d$ no full-unitary construction dominates direct access in quantum memory, classical memory and success at once, since $\lambda(N,d)>\lambda(N,D)$ for every $D>d$ and $\lambda(N,d)>q\lambda(N,2d)$; the measured-branch program can tie with it when $q=1$ but not exceed it. When $D_{\min}=d$ (first block of Table~\ref{tab:numerical}) the minimal Naimark program is the unitary PSAR compression~\cite{sedlak2019optimal} applied directly, with the same success as the direct route, and for $n=d$ nonzero rank-one effects it coincides with~\eqref{eq:direct-rank1-quantum} because then $nd=n^2$. Beyond these examples we make no optimality claim: the comparison is between explicit constructions, not between the optimal capabilities of the oracle models.

\section{Discussion}\label{sec:discussion}

\paragraph{Scope of the results.}
All constructions in this paper are achievable protocols for the oracle models of Definitions~\ref{def:direct-oracle}, \ref{def:coherent-oracle} and~\ref{def:branch-oracle}; the lower bounds of Section~\ref{sec:lower} hold for every model and do not rely on port-based teleportation. The polynomial memory of every route rests on the permutation symmetry of the stored program rather than on unitarity (Remark~\ref{proofPoly}, Lemma~\ref{lem:inv-compression}); what a representation changes is the exponent and the success probability. The comparison of Section~\ref{comparison} is between explicit constructions, not between the optimal capabilities of the oracle models, and we prove no separation between them. The unitary-PSAR success formula is optimal for unitary retrieval and its register is that of the optimal-success protocol~\cite{sedlak2019optimal}, but retrieval of a POVM only requires reproducing a fixed readout on inputs restricted to the subspace $\rho\otimes\projj0$, so this optimality does not transfer; the direct route illustrates the gap, since for the qubit SIC its compressed program needs fewer qubits than the optimal PSAR register of the minimal Naimark unitary ($23$ qubits with a $15$-bit record without any promise on the effects, and under the rank-one promise $16$ qubits with the same record or $28$ qubits with none, against $48$ at $N=50$) at a higher success probability. Since the direct program already compresses to polynomial size under black-box access, a structured representation is worth constructing only for what it changes - the degree of the polynomial memory and the division between quantum and classical memory - at a success probability that, among the protocols analysed here and for the same number of queries, is no higher than that of the direct route.

\paragraph{Assumptions on oracle access.}
We do not address how $\Omega_{\mathrm{Nai}}$ or $\Omega_{\mathrm{Br}}$ could be obtained from black-box access. Producing the Naimark unitary requires a concrete unitary extension of the dilation isometry, implemented coherently on the whole dilation space; statistical access to $\PP_{\mathbf M}$ provides neither the isometry nor its extension, and a dilated experimental setup realizes $\Omega_{\mathrm{Nai}}$ only if its ancillary input can be controlled. The Kotowski-Oszmaniec theorem~\cite{kotowski2026pretty} guarantees the existence of a branch ensemble but gives no efficient construction from black-box access. The structured models therefore isolate the storage cost of a given representation and are not end-to-end learning protocols.

\paragraph{Program creation versus storage and retrieval.}\label{sec:disc-programs}
The direct and Naimark routes are storage and retrieval of an unknown resource: the storage device never holds a description of $\mathbf M$. The branch route is different: in the scenario considered here a designer supplies a preprocessed ensemble, and obtaining it from black-box access is outside the model, while the parties who later run the program reproduce its statistics without a description of it. The branch route therefore solves a \emph{measurement-program creation} task - encoding a preprocessed measurement in a compact, remotely executable form with a self-contained classical record (Proposition~\ref{prop:finite-readout}) - and not black-box storage and retrieval under purely statistical access. This asymmetry between the preparing and the executing party is absent from the other two routes.

\paragraph{Structured POVMs.}
The factor $q=\tfrac18$ is a dimension-independent worst-case guarantee, improved to $\tfrac1d$ for $d<8$ by the no-ancilla decomposition of Remark~\ref{rem:spectral-branch}, and loose for structured measurements. For the qubit SIC POVM the two-branch decomposition of Appendix~\ref{app:explicit-KO} gives $q^{\star}=\frac{3-\sqrt3}{2}\approx0.634$, and since the SIC is rank-one with $D_{\min}=2d=4$, a single branch in $\UU(4)$ realizes it with $q=1$. For that single branch the raw readout of Proposition~\ref{prop:measured-branch} is the SIC itself and the relabelling is the identity, a fixed table that need not be stored, so the measured-branch program coincides with the direct rank-one program in register, record and success $\lambda(N,2)$; the strict gaps below concern the full-unitary branch program. Within the two-branch ansatz the ratio of heralded probabilities is $P_{\mathrm{Br}}/P_{\mathrm{Dir}}=q^{\star}\,\frac{N+3}{N+15}<q^{\star}<1$, so the branch route's success stays below the direct one for every $N$.

\paragraph{Extensions and open problems.}
The framework suggests extensions to unital channels via the randomized simulation of~\cite{kotowski2026randomized} and to general channels via a Stinespring-dilation oracle; their access models and resource formulas are outside the scope of this paper. Open problems that remain to be addressed are: the optimal success probability attainable at a given memory under each oracle, in particular whether $\lambda(N,d)$ can be improved under direct access and whether a dimension-independent universal $q>\tfrac18$ exists~\cite{kotowski2026pretty}; an unconditional lower bound matching the memory~\eqref{eq:direct-cq-memory} of the compressed direct program, for which we claim no optimality; the constants in the $\Theta(\log n)$ and $\Theta\bigl(\log\frac1{1-p}\bigr)$ scalings, the latter inherited from~\cite{perezgarcia2006optimality} and informative only very close to $p=1$ (Remark~\ref{rem:success-lower}); and the optimal dependence on the input dimension.

\section{Conclusions}
Exact heralded storage and retrieval of an unknown measurement needs only polynomial memory under black-box access: the program of the universal port-based protocol compresses exactly into the type of the outcome string and a quantum register, while storing only the type, without re-aligning the reference systems, does not permit exact retrieval by a decoder that returns an outcome on every success. For every access model the total memory must grow at least as $\log_2 n$ with the number of outcomes, must satisfy $2Q+C\ge\log_2\bigl((n-1)d^2+1\bigr)$, and must grow at least logarithmically in $1/(1-p)$ as the success probability $p\to1$; the compressed black-box program attains the first and the last of these bounds up to constants, so the total memory is $\Theta(\log n)$ at fixed $d$ and target success below one, and $\Theta\bigl(\log\frac1{1-p}\bigr)$ at fixed $d,n\ge2$. The structured representations change the success probability, the degree of the polynomial memory and its quantum/classical division: a Naimark dilation gives a purely quantum program governed by the dilation dimension, which is never smaller or more successful than the program obtained by measuring the dilation readout during storage, and a branch program a quantum register independent of the number of outcomes together with a sampled table of $2d$ labels independent of the number of queries, at a success probability below the non-failure probability of the decomposition; measuring the branch during storage reduces that register and raises the success at the same table. Because the branch route assumes a designer who supplies the decomposition, it is measurement-program creation rather than learning of an unknown device. Neither the optimal quantum-memory cost of measurement retrieval nor a separation between the oracle models is settled in this paper.

\section*{Author contributions}

FT developed the theoretical framework, formulated the three oracle models, and proved the main results on exact heralded storage and retrieval of POVMs. ZP supervised the project, contributed to the interpretation of the results within the representation-dependent framework, and guided the overall direction. LM performed the numerical computations, generated the supporting figures and tables, and verified the numerical claims. Large language models (Claude by Anthropic, and OpenAI models used to produce internal adversarial reviews) were used in preparing this manuscript: for organisation, LaTeX formatting and clarity improvements, and also for drafting candidate mathematical statements and proof sketches, including suggestions arising from AI-assisted internal reviews of earlier drafts. Every statement and proof in the paper was checked, and where necessary corrected, by the authors, who take full responsibility for its content.

\section*{Acknowledgments}
FT acknowledges support from the National Science Center (NCN), Poland, under Project Sonata Bis 10, No. 2020/38/E/ST3/00269. LM and ZP acknowledge support from the National Science Center (NCN), Poland, under Project Opus No. 2022/47/B/ST6/02380.

\section*{Data and code availability}

All register formulas and tabulated values are reproducible from \texttt{calculations.py} in the supplementary materials, which evaluates $K_d$ through Littlewood's identity~\cite[Ch.~I.5]{macdonald1995symmetric} and maximizes over outcome types exactly. Code and data are available at the accompanying repository: \url{https://github.com/iitis/Quantum_Measurement_Protocols}.
\section*{Competing interests}

The authors declare no competing financial interests.

\bibliographystyle{quantum}
\bibliography{Bibliography}

@book{watrous2018theory,
  title={The theory of quantum information},
  author={Watrous, John},
  year={2018},
  publisher={Cambridge university press},
  doi={10.1017/9781316848142}
}

@book{bengtsson2017,
    title={Geometry of Quantum States: An Introduction to Quantum Entanglement},
    author={Bengtsson, Ingemar and Życzkowski, Karol},
    year={2017},
    publisher={Cambridge university press},
    doi={10.1017/CBO9780511535048}
}

@article{lewandowska2022storage,
  title={Storage and retrieval of von Neumann measurements},
  author={Lewandowska, Paulina and Kukulski, Ryszard and Pawela, {\L}ukasz and Pucha{\l}a, Zbigniew},
  journal={Physical Review A},
  volume={106},
  number={5},
  pages={052423},
  year={2022},
  publisher={APS},
  doi={10.1103/PhysRevA.106.052423}
}

@article{christandl2021asymptotic,
  title={Asymptotic performance of port-based teleportation},
  author={Christandl, Matthias and Leditzky, Felix and Majenz, Christian and Smith, Graeme and Speelman, Florian and Walter, Michael},
  journal={Communications in Mathematical Physics},
  volume={381},
  number={1},
  pages={379--451},
  year={2021},
  publisher={Springer},
  doi={10.1007/s00220-020-03884-0}
}

@article{studzinski2017port,
  title={Port-based teleportation in arbitrary dimension},
  author={Studzi{\'n}ski, Micha{\l} and Strelchuk, Sergii and Mozrzymas, Marek and Horodecki, Micha{\l}},
  journal={Scientific reports},
  volume={7},
  number={1},
  pages={10871},
  year={2017},
  publisher={Nature Publishing Group UK London},
  doi={10.1038/s41598-017-10051-4}
}

@article{sedlak2019optimal,
  title={Optimal probabilistic storage and retrieval of unitary channels},
  author={Sedl{\'a}k, Michal and Bisio, Alessandro and Ziman, M{\'a}rio},
  journal={Physical review letters},
  volume={122},
  number={17},
  pages={170502},
  year={2019},
  publisher={APS},
  doi={10.1103/PhysRevLett.122.170502}
}

@article{kotowski2026pretty,
  title={Pretty good simulation of all quantum measurements by projective measurements},
  author={Kotowski, Micha{\l} and Oszmaniec, Micha{\l}},
  journal={IEEE Transactions on Information Theory},
  volume={72},
  number={5},
  pages={3156-3174},
  year={2026},
  publisher={IEEE},
  doi={10.1109/TIT.2026.3678498}
}

@article{oszmaniec2018simulating,
  title={Simulating all quantum measurements using only projective measurements and postselection},
  author={Oszmaniec, Micha{\l} and Maciejewski, Filip B and Pucha{\l}a, Zbigniew},
  journal={Phys. Rev. A},
  volume={100},
  pages={012351},
  year={2019},
  doi={10.1103/PhysRevA.100.012351}
}

@article{yoshida2026one,
  title={One-to-one correspondence between deterministic port-based teleportation and unitary estimation},
  author={Yoshida, Satoshi and Koizumi, Yuki and Studzi{\'n}ski, Micha{\l} and Quintino, Marco T{\'u}lio and Murao, Mio},
  journal={IEEE Transactions on Information Theory},
  year={2026},
  publisher={IEEE},
  doi={10.1109/TIT.2026.3658543}
}

@article{kotowski2026randomized,
  title={Randomized simulation of quantum channels using small ancilla},
  author={Kotowski, Marcin and Kotowski, Micha{\l}},
  journal={arXiv preprint arXiv:2606.08784},
  year={2026},
  doi={10.48550/arXiv.2606.08784}
}

@article{Bisio2011,
  title={Quantum networks: general theory and applications},
  author={Bisio, Alessandro and Chiribella, Giulio and D'Ariano, Giacomo Mauro and Perinotti, Paolo},
  journal={Acta Physica Slovaca},
  volume={61},
  number={3},
  pages={273--390},
  year={2011},
  doi={10.2478/v10155-011-0003-9},
  eprint={1601.04864},
  archivePrefix={arXiv}
}

@article{renes2004symmetric,
  title={Symmetric informationally complete quantum measurements},
  author={Renes, Joseph M and Blume-Kohout, Robin and Scott, Andrew J and Caves, Carlton M},
  journal={Journal of Mathematical Physics},
  volume={45},
  number={6},
  pages={2171--2180},
  year={2004},
  publisher={American Institute of Physics},
  doi={10.1063/1.1737053}
}

@article{bisio2011quantum,
  title={Quantum learning algorithms for quantum measurements},
  author={Bisio, Alessandro and D'Ariano, Giacomo Mauro and Perinotti, Paolo and Sedl{\'a}k, Michal},
  journal={Physics Letters A},
  volume={375},
  number={39},
  pages={3425--3434},
  year={2011},
  publisher={Elsevier},
  doi={10.1016/j.physleta.2011.08.002}
}

@article{bisio2010optimal,
  title={Optimal quantum learning of a unitary transformation},
  author={Bisio, Alessandro and Chiribella, Giulio and D’Ariano, Giacomo Mauro and Facchini, Stefano and Perinotti, Paolo},
  journal={Physical Review A -- Atomic, Molecular, and Optical Physics},
  volume={81},
  number={3},
  pages={032324},
  year={2010},
  publisher={APS},
  doi={10.1103/PhysRevA.81.032324}
}

@article{ishizaka2008asymptotic,
  title={Asymptotic teleportation scheme as a universal programmable quantum processor},
  author={Ishizaka, Satoshi and Hiroshima, Tohya},
  journal={Physical review letters},
  volume={101},
  number={24},
  pages={240501},
  year={2008},
  publisher={APS},
  doi={10.1103/PhysRevLett.101.240501}
}

@article{ishizaka2009quantum,
  title={Quantum teleportation scheme by selecting one of multiple output ports},
  author={Ishizaka, Satoshi and Hiroshima, Tohya},
  journal={Physical Review A -- Atomic, Molecular, and Optical Physics},
  volume={79},
  number={4},
  pages={042306},
  year={2009},
  publisher={APS},
  doi={10.1103/PhysRevA.79.042306}
}

@article{jamiolkowski1972linear,
  title={Linear transformations which preserve trace and positive semidefiniteness of operators},
  author={Jamio{\l}kowski, Andrzej},
  journal={Reports on mathematical physics},
  volume={3},
  number={4},
  pages={275--278},
  year={1972},
  publisher={Elsevier},
  doi={10.1016/0034-4877(72)90011-0}
}

@article{choi1975completely,
  title={Completely positive linear maps on complex matrices},
  author={Choi, Man-Duen},
  journal={Linear algebra and its applications},
  volume={10},
  number={3},
  pages={285--290},
  year={1975},
  publisher={Elsevier},
  doi={10.1016/0024-3795(75)90075-0}
}

@article{uhlmann1976transition,
  title={The “transition probability” in the state space of a*-algebra},
  author={Uhlmann, Armin},
  journal={Reports on Mathematical Physics},
  volume={9},
  number={2},
  pages={273--279},
  year={1976},
  publisher={Elsevier},
  doi={10.1016/0034-4877(76)90060-4}
}

@article{jozsa1994fidelity,
  title={Fidelity for mixed quantum states},
  author={Jozsa, Richard},
  journal={Journal of modern optics},
  volume={41},
  number={12},
  pages={2315--2323},
  year={1994},
  publisher={Taylor \& Francis},
  doi={10.1080/09500349414552171}
}

@article{chefles2000quantum,
  title={Quantum state discrimination},
  author={Chefles, Anthony},
  journal={Contemporary Physics},
  volume={41},
  number={6},
  pages={401--424},
  year={2000},
  publisher={Taylor \& Francis},
  doi ={10.1080/00107510010002599}
}

@inproceedings{Aaronson2017Shadow,
  title={Shadow tomography of quantum states},
  author={Aaronson, Scott},
  booktitle={Proceedings of the 50th annual ACM SIGACT symposium on theory of computing},
  pages={325--338},
  year={2018},
  doi={10.1145/3188745.3188802}
}

@article{markiewicz2023duality,
  title={Duality of averaging of quantum states over arbitrary symmetry groups revealing Schur--Weyl duality},
  author={Markiewicz, Marcin and Przewocki, Janusz},
  journal={Journal of Physics A: Mathematical and Theoretical},
  volume={56},
  number={39},
  pages={395301},
  publisher={IOP Publishing},
  year={2023},
  doi = {10.1088/1751-8121/acf4d5}
}

@article{naimarkblessedpaper,
  title={On a representation of additive operator set functions},
  author={Naimark, Mark A.},
  journal={Dokl. Akad. Nauk SSSR},
  volume={41},
  number={9},
  pages={373--375},
  year={1943},
  nolink={}
}

@article{Pirandola2019,
  title={Fundamental limits to quantum channel discrimination},
  author={Pirandola, Stefano and Laurenza, Riccardo and Lupo, Cosmo and Pereira, Jason L},
  journal={npj Quantum Information},
  volume={5},
  number={1},
  pages={50},
  year={2019},
  publisher={Nature Publishing Group UK London},
  doi={10.1038/s41534-019-0162-y}
}

@article{wang2016higher,
  title={Higher-dimensional performance of port-based teleportation},
  author={Wang, Zhi-Wei and Braunstein, Samuel L},
  journal={Scientific Reports},
  volume={6},
  number={1},
  pages={33004},
  doi = {10.1038/srep33004},
  year={2016},
  publisher={Nature Publishing Group UK London}
}

@article{Chen2007,
  title = {Ancilla dimensions needed to carry out positive-operator-valued measurement},
  author = {Chen, Ping-Xing and Bergou, J\'anos A. and Zhu, Shi-Yao and Guo, Guang-Can},
  journal = {Phys. Rev. A},
  volume = {76},
  issue = {6},
  pages = {060303(R)},
  numpages = {4},
  year = {2007},
  month = {Dec},
  publisher = {American Physical Society},
  doi = {10.1103/PhysRevA.76.060303},
}

@article{nielsen1997programmable,
  title={Programmable quantum gate arrays},
  author={Nielsen, Michael A. and Chuang, Isaac L.},
  journal={Physical review letters},
  volume={79},
  number={2},
  pages={321},
  year={1997},
  publisher={APS},
  doi={10.1103/PhysRevLett.79.321}
}

@article{perezgarcia2006optimality,
  title={Optimality of programmable quantum measurements},
  author={P{\'e}rez-Garc{\'\i}a, David},
  journal={Physical Review A},
  volume={73},
  number={5},
  pages={052315},
  year={2006},
  publisher={APS},
  doi={10.1103/PhysRevA.73.052315}
}

@article{yang2016efficient,
  title={Efficient Quantum Compression for Ensembles of Identically Prepared Mixed States},
  author={Yang, Yuxiang and Chiribella, Giulio and Ebler, Daniel},
  journal={Physical Review Letters},
  volume={116},
  number={8},
  pages={080501},
  year={2016},
  doi={10.1103/PhysRevLett.116.080501}
}

@article{fiurasek2004probabilistic,
  title={Probabilistic quantum multimeters},
  author={Fiur{\'a}{\v{s}}ek, Jarom{\'\i}r and Du{\v{s}}ek, Miloslav},
  journal={Physical Review A},
  volume={69},
  number={3},
  pages={032302},
  year={2004},
  doi={10.1103/PhysRevA.69.032302}
}

@article{bluhm2018compression,
  title={Quantum compression relative to a set of measurements},
  author={Bluhm, Andreas and Rauber, Lukas and Wolf, Michael M.},
  journal={Annales Henri Poincar{\'e}},
  volume={19},
  number={6},
  pages={1891--1937},
  year={2018},
  doi={10.1007/s00023-018-0660-z}
}

@article{kubicki2019resource,
  title={Resource Quantification for the No-Programing Theorem},
  author={Kubicki, Aleksander M. and Palazuelos, Carlos and P{\'e}rez-Garc{\'\i}a, David},
  journal={Physical Review Letters},
  volume={122},
  number={8},
  pages={080505},
  year={2019},
  doi={10.1103/PhysRevLett.122.080505}
}

@book{macdonald1995symmetric,
  title={Symmetric Functions and Hall Polynomials},
  author={Macdonald, Ian G.},
  edition={2},
  publisher={Oxford University Press},
  address={Oxford},
  year={1995},
  nolink={}
}

\appendix
\section{Proof of Proposition~\ref{prop:direct-rank1}}
\label{app:rank1-proof}

\begin{proof}[Proof of Proposition~\ref{prop:direct-rank1}]
\emph{Conditional operators.} With $\ket{\psi^+}=d^{-1/2}\sum_i\ket{ii}$ normalized, 
\[
\Tr_B\bigl[(\id\otimes X)\projj{\psi^+}\bigr]=X^{T}/d
\] for every $X\in\BB(\C^d)$, the transpose being taken in the Schmidt basis of $\ket{\psi^+}$. Since $\omega_{AB}=(O_A\otimes\id)\projj{\psi^+}^{\otimes N}(O_A\otimes\id)$ with $O_A$ acting on $A^N$ only, and $\PP_{\mathbf M}^{\otimes N}$ acts on $B^N$, the sub-normalized post-measurement operator on $A^N$ for the outcome word $o$ is
\begin{equation}\label{eq:rho-o-explicit}
\rho_{\mathbf M}(o)=\Tr_{B^N}\Bigl[\bigl(\id\otimes\textstyle\bigotimes_kM_{o_k}\bigr)\omega_{AB}\Bigr]=d^{-N}\,O_A\Bigl(\bigotimes_{k=1}^NM_{o_k}^{T}\Bigr)O_A .
\end{equation}
\emph{Rank-one support.} If $M_j=\projj{\phi_j}$ or $M_j=0$ for every $j$, then $\bigotimes_kM_{o_k}^T=\projj{v_o}$ with $\ket{v_o}=\bigotimes_k\ket{\bar\phi_{o_k}}$, where $\ket{v_o}=0$ if some $M_{o_k}=0$; such words have probability zero and need not be stored. For the sorted word $o^t$ of type $t$, $\ket{v_{o^t}}=\bigotimes_{j=1}^n\ket{\bar\phi_j}^{\otimes t_j}$ lies in $\mathcal S_t:=\bigotimes_{j=1}^n\mathrm{Sym}^{t_j}(\C^d)$.

\emph{$O_A$ preserves $\mathcal S_t$.} Let $S_t=S_{t_1}\times\dots\times S_{t_n}\subseteq S_N$ be the stabilizer of $o^t$, acting on the $j$-th block of $t_j$ tensor factors. The subspace of $(\C^d)^{\otimes N}$ fixed by every $V_\pi$, $\pi\in S_t$, is exactly $\mathcal S_t$, because for groups $G_1,G_2$ acting on $V_1,V_2$ one has 
\[
\mathrm{Fix}(G_1\times G_2)=\mathrm{Fix}(G_1\times\{e\})\cap\mathrm{Fix}(\{e\}\times G_2)=(\mathrm{Fix}\,G_1\otimes V_2)\cap(V_1\otimes\mathrm{Fix}\,G_2)=\mathrm{Fix}\,G_1\otimes\mathrm{Fix}\,G_2,
\] (the last step by choosing bases of $V_1,V_2$ adapted to the two subspaces), iterated over the $n$ blocks, and the fixed points of $S_m$ on $(\C^d)^{\otimes m}$ form $\mathrm{Sym}^m(\C^d)$. By Lemma~\ref{lem:direct-symmetry}, $[O_A,V_\pi]=0$ for all $\pi\in S_N$; hence for $\ket w\in\mathcal S_t$ and $\pi\in S_t$, $V_\pi O_A\ket w=O_AV_\pi\ket w=O_A\ket w$, i.e.\ $O_A\ket w\in\mathcal S_t$. Therefore $\rho_{\mathbf M}(o^t)=d^{-N}\,O_A\projj{v_{o^t}}O_A$ is supported on $\mathcal S_t$, a subspace independent of $\mathbf M$ with $\dim\mathcal S_t=\prod_j\dim\mathrm{Sym}^{t_j}(\C^d)=\prod_j\binom{t_j+d-1}{d-1}$.

\emph{Encoder and decoder.} For each type $t$ fix an isometry $W_t:\mathcal S_t\to\C^{\dim\mathcal S_t}$. Encoding: read $o$, apply $V^{-1}_{\pi_o}$ to $A^N$, which by~\eqref{eq:direct-covariance} turns $\rho_{\mathbf M}(o)$ into $\rho_{\mathbf M}(o^t)$, then apply $W_t$; store $t$ and the register. Decoding: apply $W_t^\dagger$ and then the randomized decoder of Proposition~\ref{prop:direct-cq}. All maps are independent of $\mathbf M$ and reproduce $\sigma_{\mathbf M}$ exactly, so retrieval and its success probability $\lambda(N,d)$ are unchanged.

\emph{Quantum type register.} The stored program is the classical--quantum state $\sum_tp(t)\projj{t}\otimes\hat\rho_t$ with $\hat\rho_t$ supported on $\mathcal S_t$. Let $J_t$ be the inclusion of $\mathcal S_t$ into $\bigoplus_{t'\in T_{n,N}}\mathcal S_{t'}$ and store instead $\sum_tp(t)\,J_t\hat\rho_tJ_t^\dagger$, a state on a space of dimension $\sum_t\dim\mathcal S_t$. The projectors $J_tJ_t^\dagger$ are independent of $\mathbf M$ and the stored state is block diagonal with respect to them, so the receiver recovers $t$ by measuring the sector without disturbing the block and then applies the decoder above; this gives~\eqref{eq:direct-rank1-quantum}. Since $\sum_t\dim\mathcal S_t\le|T_{n,N}|\max_t\dim\mathcal S_t$, the purely quantum program satisfies $Q^{\mathrm q}_{\mathrm{rank}\text{-}1}\le Q_{\mathrm{rank}\text{-}1}+C$ while using more qubits than the hybrid one. The same packing applies to Proposition~\ref{prop:direct-cq}, with total dimension $\sum_t\prod_jK_d(t_j)$.

\emph{Closed form.} $\sum_{t\in T_{n,N}}\prod_j\binom{t_j+d-1}{d-1}=[x^N]\prod_{j=1}^n(1-x)^{-d}=[x^N](1-x)^{-nd}=\binom{N+nd-1}{N}$, since $\sum_{m\ge0}\binom{m+d-1}{d-1}x^m=(1-x)^{-d}$.
\end{proof}

\section{An explicit Oszmaniec-style decomposition for the qubit SIC}
\label{app:explicit-KO}

The qubit SIC POVM admits a simple closed-form two-branch decomposition, each branch carrying two SIC outcomes and one failure outcome; the heralding factor $q$ is limited only by positivity of the failure effects.

Let $\{ M_i = \frac12 \projj{\psi_i} \}_{i=1}^4 $ denote the qubit SIC POVM. The SIC relations imply 
\begin{equation}
    \sum_i \projj{\psi_i} = 2 \id_{2}, \quad \text{and} \quad |\langle\psi_{i}|\psi_{j}\rangle|^{2}=\tfrac{1}{3},\quad \text{for} \quad i\neq j.
\end{equation}
Partition the four SIC outcomes into two pairs, with weights $p_{1}=p_{2}=\tfrac{1}{2}$:
\begin{align*}
\mathbf{N}^{(1)}&=\bigl(\;q\,\projj{\psi_{1}},\;q\,\projj{\psi_{2}},\;\id_{2}-q\bigl(\projj{\psi_{1}}+\projj{\psi_{2}}\bigr)\;\bigr)\quad\text{on outcomes}\;\{1,2,\emptyset\},\\
\mathbf{N}^{(2)}&=\bigl(\;q\,\projj{\psi_{3}},\;q\,\projj{\psi_{4}},\;\id_{2}-q\bigl(\projj{\psi_{3}}+\projj{\psi_{4}}\bigr)\;\bigr)\quad\text{on outcomes}\;\{3,4,\emptyset\}.
\end{align*}
The corresponding classical relabellings are $\mathcal{Q}^{(1)}\!:\{1,2,3,4\}\to\{1,2,\emptyset,\emptyset\}$ and $\mathcal{Q}^{(2)}\!:\{1,2,3,4\}\to\{3,4,\emptyset,\emptyset\}$ when each branch is realised as a rank-1 projective measurement on $\C^{4}$.

\begin{propo}
    The above two-branch construction defines a valid Kotowski-Oszmaniec decomposition identity for every $0 < q \le q^{\star}$, where
    \begin{equation}
        q^{\star}=\frac{3-\sqrt{3}}{2}.
    \end{equation}
\end{propo}
\begin{proof}
    For the simulated outcomes $i\in\{1,\dots,4\}$,
    \begin{equation}
        \sum_{\beta}p_{\beta}\,N^{(\beta)}_{i}\;=\;\tfrac{1}{2}\,q\,\projj{\psi_{i}}\;=\;q\,M_{i},
    \end{equation}
    since $M_{i}=\tfrac{1}{2}\projj{\psi_{i}}$ and exactly one branch contributes to each outcome.  For the failure outcome,
    \begin{equation}
        \sum_{\beta}p_{\beta}\,N^{(\beta)}_{\emptyset} =\tfrac{1}{2}\Bigl[2\,\id_{2}-q\sum_{i=1}^{4}\projj{\psi_{i}}\Bigr] =\id_{2}-q\cdot\id_{2}=(1-q)\,\id_{2},
    \end{equation}
    where we used the SIC tight-frame identity $\sum_{i=1}^{4}\projj{\psi_{i}}=2\,\id_{2}$ (equivalent to POVM completeness $\sum_{i}M_{i}=\id_{2}$). Thus the decomposition identity holds for every value of $q$. The remaining constraint comes from the requirement that each branch be a valid POVM. The only nontrivial positivity condition concerns the failure effect $N^{(\beta)}_{\emptyset}\succeq 0$. For the first branch this is equivalent to
    \begin{equation}
        q(\projj{\psi_1}+\projj{\psi_2}) \preceq \id_2.
    \end{equation}
    The sum of two rank-1 projectors with overlap $|\langle\psi_{1}|\psi_{2}\rangle|=\frac{1}{\sqrt{3}}$ has spectrum $\{\,1\pm \frac{1}{\sqrt{3}}\,\}$, giving the operator norm $1+\frac{1}{\sqrt{3}}=\frac{\sqrt{3}+1}{\sqrt{3}}$.  Hence
    \begin{equation}
        q^{\star}=\frac{\sqrt{3}}{\sqrt{3}+1}=\frac{3-\sqrt{3}}{2}\approx 0.6340
    \end{equation}
    and the two-branch ansatz is valid for all $q\in(0,q^{\star}]$.  Thus, for the qubit SIC POVM, the explicit construction yields a heralding factor $q^{\star} \approx 5.07 \cdot \frac18$, substantially improving over the universal worst-case guarantee. 
\end{proof}
\begin{remark}
The value $q^{\star}=\tfrac{3-\sqrt3}{2}$ is optimal only within the two-branch ansatz above, in which each branch simulates exactly two of the four SIC outcomes and routes its remaining readout outcomes to failure. This restriction is essential: the qubit SIC is rank-one, so its minimal Naimark dilation has dimension $D_{\min}=n=4$ (Remark~\ref{rem:naimark-dim}), which fits within the branch space $\C^{2d}=\C^{4}$ of a single branch. A single branch $U^{(\beta)}\in\UU(4)$ therefore already realizes the SIC \emph{exactly}, with $q=1$ and no failure outcome. Consequently $q^{\star}$ is not globally optimal in the unrestricted branch model; the two-branch construction is of value as an explicit, closed-form decomposition under a bounded number of successful outcomes per branch.
\end{remark}

Whether a universal guarantee $q>\tfrac18$ holds for all POVMs remains open~\cite{kotowski2026pretty}; the example shows that for concrete POVMs the universal bound is far from tight.

\end{document}